%% file: main_IEEE.tex
\documentclass[journal]{IEEEtran}

\input{setup/packages}

\input{setup/macros}

\begin{document}

\title{Privacy-preserving Searchable Databases \\ with Controllable Leakage*\thanks{*This article is an extension of our initial work appeared in the proceedings of IEEE 10th International Conference on Cloud Computing (CLOUD) 2017 under the title ``P-McDb: Privacy-preserving Search using Multi-cloud Encrypted Databases'' by Shujie Cui, Muhammad Rizwan Asghar, Steven D Galbraith, and Giovanni Russello \cite{CuiAGR17}.}}

\author{Shujie Cui, Xiangfu Song, Muhammad Rizwan Asghar, Steven D Galbraith, and Giovanni Russello

\thanks{Shujie Cui is with the Large-Scale Data \& Systems (LSDS) group in the Department of Computing, Imperial College London, UK.

Xiangfu Song is with the School of Computer Science and Technology, Shandong University, Jinan, China.

Muhammad Rizwan Asghar, Steven D Galbraith, and Giovanni Russello are with the Cyber Security Foundry, The University of Auckland, New Zealand.

They can be contacted by email: s.cui@imperial.ac.uk, bintasong@gmail.com, r.asghar@auckland.ac.nz, s.galbraith@auckland.ac.nz, and g.russello@auckland.ac.nz, respectively.}
}

\maketitle

\input{abstract}


\input{sections/intro}

\input{sections/notation}
\input{sections/leakage}

\input{sections/attacks}

\input{sections/overview}

\input{sections/details}

\input{sections/proof}

\input{sections/performance}

\input{sections/conclusion}

\bibliographystyle{ieeetr}
\bibliography{SSE}

\end{document}

%% file: setup/packages.tex
\usepackage{amsmath}
\usepackage{amssymb}         
\usepackage{amsthm}
\usepackage{graphicx}        
\usepackage{footnote}
\usepackage[hang,flushmargin]{footmisc}
\usepackage{marvosym}
\usepackage{epstopdf}
\usepackage{epsfig}
\usepackage{algorithm}
\usepackage[noend]{algorithmic}
\usepackage{diagbox}
\usepackage{subfigure}
\usepackage{hyperref}
\usepackage{todonotes}
\usepackage{color}
\usepackage{bm} 
\usepackage{multirow}

\usepackage{twoopt}
\usepackage{tikz}
\usetikzlibrary{fit}

\ifCLASSOPTIONcompsoc
  \usepackage[nocompress]{cite}
\else
  \usepackage{cite}
\fi

%% file: setup/macros.tex
\usepackage{xspace}
\newcommand{\etal}{\textit{et al.}\xspace}

\newcommand{\ie}{\textit{i.e.,}\xspace}
\newcommand{\eg}{\textit{e.g.,}\xspace}

\mathchardef\mhyphen="2D

\newcommand{\rt}{\rightarrow}
\newcommand{\lt}{\leftarrow}

\newcommand{\framework}{\textsl{\mbox{P-McDb}}\xspace}

\newcommand{\tgrey}[1]{\textcolor{lightgray}{#1}}

\renewcommand{\paragraph}[1]{\medskip \noindent \textbf{#1.\ }}

\newtheorem{mydef}{Definition}
\newtheorem{theorem}{Theorem}

%% file: abstract.tex
\begin{abstract}
Searchable Encryption (SE) is a technique that allows Cloud Service Providers (CSPs) to search over encrypted datasets without learning the content of queries and records.
In recent years, many SE schemes have been proposed to protect outsourced data from CSPs.
Unfortunately, most of them leak sensitive information, from which the CSPs could still infer the content of queries and records by mounting leakage-based inference attacks, such as the \emph{count attack} and \emph{file injection attack}.

In this work, first we define the leakage in searchable encrypted databases and analyse how the leakage is leveraged in existing leakage-based attacks.
Second, we propose a \underline{P}rivacy-preserving \underline{M}ulti-\underline{c}loud based dynamic symmetric SE (SSE) scheme for relational \underline{D}ata\underline{b}ase (\framework).
\framework has minimal leakage, which not only ensures confidentiality of queries and records, but also protects the search, access, and size patterns from CSPs.
Moreover, \framework ensures both forward and backward privacy of the database.
Thus, \framework could resist existing leakage-based attacks, \eg active file/record-injection attacks.
We give security definition and analysis to show how \framework hides the aforementioned patterns.
Finally, we implemented a prototype of \framework and test it using the TPC-H benchmark dataset.
Our evaluation results show the feasibility and practical efficiency of \framework.
\end{abstract}

%% file: sections/intro.tex
\section{Introduction}
\label{sec:introduction}
Cloud computing is a successful paradigm offering companies and individuals virtually unlimited
data storage and computational power at very attractive costs.
However, uploading sensitive data, such as medical, social, and financial information, to public cloud environments is still a challenging issue due to security concerns.
In particular, once such data sets and related operations are uploaded to cloud environments, the tenants must therefore trust the Cloud Service Providers (CSPs).
Yet, due to possible cloud infrastructure bugs~\cite{GunawiHLPDAELLM14}, misconfigurations~\cite{dropboxleaks} and external attacks~\cite{verizonreport}, the data could be disclosed or corrupted.
Searchable Encryption (SE) is an effective approach that allows organisations to outsource their databases and search operations to untrusted CSPs, without compromising the confidentiality of records and queries.

Since the seminal SE paper by Song \etal \cite{Song:2000:Practical}, a long line of work has investigated SE schemes with flexible functionality and better performance \cite{Asghar:2013:CCSW,Curtmola:2006:Searchable,Popa:2011:Cryptdb,Sarfraz:2015:DBMask}.
These schemes are proved to be secure in certain models under various cryptographic assumptions.
Unfortunately, a series of more recent work \cite{Islam:2012:Access,Naveed:2015:Inference,Cash:2015:leakage,Zhang:2016:All,Kellaris:ccs16:Generic,Abdelraheem:eprint17:record}
illustrates that they are still vulnerable to inference attacks, where malicious CSPs could recover the content of queries and records by (i) observing the data directly from the encrypted database and (ii) learning about the results and queries when users access the database.

From the encrypted database, the CSP might learn the frequency information of the data.
From the search operation, the CSP is able to know the \emph{access pattern}, \ie the records returned to users in response to given queries.
The CSP can also infer if two or more queries are equivalent, referred to as the \emph{search pattern}, by comparing the encrypted queries or matched data.
Last but not least, the CSP can simply log the number of matched records or files returned by each query, referred to as the \emph{size pattern}.

When an SE scheme supports insert and delete operations, it is referred to as a \emph{dynamic} SE scheme.
Dynamic SE schemes might leak extra information if they do not support \emph{forward privacy} and \emph{backward privacy} properties.
Lacking forward privacy means that the CSP can learn if newly inserted data or updated data matches previously executed queries.
Missing backward privacy means that the CSP learns if deleted data matches new queries.
Supporting forward and backward privacy is fundamental to limit the power of the CSP to collect information on how the data evolves over time.
However, only a few schemes \cite{BostMO17,ZuoSLSP18,ChamaniPPJ18,AmjadKM19} ensure both properties simultaneously.

Initiated by Islam \etal (IKK)\cite{Islam:2012:Access}, more recent works \cite{Cash:2015:leakage,Naveed:2015:Inference,Zhang:2016:All,Kellaris:ccs16:Generic} have shown that such leakage can be exploited to learn sensitive information and break the scheme.
Naveed \etal \cite{Naveed:2015:Inference} recover more than $60\%$ of the data in CryptDB \cite{Popa:2011:Cryptdb} using frequency analysis only.
Zhang \etal \cite{Zhang:2016:All} further investigate the consequences of leakage by injecting chosen files into the encrypted storage.
Based on the access pattern, they could recover a very high fraction of searched keywords by injecting a small number of known files.
Cash \etal \cite{Cash:2015:leakage} give a comprehensive analysis of the leakage in SE solutions for file collection and introduced the \emph{count attack}, where an adversary could recover queries by counting the number of matched records even if the encrypted records are semantically secure.


In this article, we investigate the leakage and attacks against relational databases\footnote{In the rest of this article, we use the term \emph{database} to refer to a relational database.} and present a \underline{P}rivacy-preserving \underline{M}ulti-\underline{c}loud based dynamic SSE scheme for \underline{D}ata\underline{b}ases (\framework).
\framework can effectively resist attacks based on the search, size or/and access patterns.
Our key technique is to use three non-colluding cloud servers: one server stores the data and performs the search operation, and the other two manage re-randomisation and shuffling of the database for protecting the access pattern.
A user with access to all servers can perform an encrypted search without leaking the search, access, or size pattern.
When updating the database, \framework also ensures both forward and backward privacy.
We give full proof of security against honest-but-curious adversaries and show how \framework can hide these patterns effectively.


The contributions of this article can be summarised as follows:
\begin{itemize}

    \item 
    We provide leakage definition specific to searchable encrypted databases, and then review how existing attacks leverage the leakage to recover queries and records. 

   \item We propose a privacy-preserving SSE database \framework, which protects the search, access, and size patterns, and achieves both forward and backward privacy, thus ensuring protection from leakage-based attacks.


  \item 
  We give full proof of security against honest-but-curious adversaries and show how \framework can effectively hide these patterns and resist leakage-based attacks.

    \item 
    Finally, we implement a prototype of \framework and show its practical efficiency by evaluating its performance on TPC-H dataset. 
\end{itemize}

The rest of this article is organised as follows.
In Section~\ref{sec:notation}, we define notations.
We present the leakage levels in SE schemes and review leakage-based attacks in Section~\ref{sec:leakage}.
In Section~\ref{sec:overview}, we provide an overview of \framework.
Solution details can be found in Section~\ref{sec:MCDB-details}.
In Section~\ref{sec:security}, we analyse the security of \framework.
Section~\ref{sec:MCDB-perf} reports the performance of \framework.
Finally, we conclude this article in Section~\ref{sec:conclusion}.

%% file: sections/notation.tex
\section{Notations and Definitions}
\label{sec:notation}
\begin{table}
\centering
\caption{Notation and description}
\scriptsize
\label{tbl:notation}
\begin{tabular}{|l|l|} \hline
\textbf{Notation}   & \textbf{Description}  \\ \hline
$e$        & Data element \\ \hline
$|e|$      &The length of data element \\ \hline
$F$        & Number of attributes or fields  \\ \hline
$rcd_{id}=(e_{id, 1}, \ldots, e_{id, F})$ & The $id$-th record  \\ \hline
$N$        & Number of records in the database  \\ \hline
$DB=\{rcd_1, \ldots, rcd_N\}$             & Database  \\ \hline
$DB(e)=\{rcd_{id} | e \in rcd_{id}\}$  & Records containing $e$ in $DB$  \\ \hline
$O(e)=|DB(e)|$                         & Occurrence of $e$ in $DB$  \\ \hline
$U_f=\cup \{e_{id, f}\}$  & The set of distinct elements in field $f$  \\ \hline
$U=\{U_{1}, ..., U_F\}$   & All the distinct elements in $DB$  \\ \hline
$e^*$      & Encrypted element  \\ \hline
$Ercd$     & Encrypted record   \\ \hline
$EDB$      & Encrypted database  \\ \hline
$Q=(type, f, e)$             & Query  \\ \hline
$Q.type$   & `select' or `delete' \\ \hline
$Q.f$      & Identifier of interested field  \\ \hline
$Q.e$      & Interested keyword   \\ \hline
$EQ$        & Encrypted query  \\ \hline
$EDB(EQ)$ or $EDB(Q)$   & Search result of $Q$  \\ \hline
$(f, g)$   & Group $g$ in field $f$  \\ \hline
$\bm{E}_{f,g}$ & Elements included in group $(f,g)$ \\ \hline
$\tau_{f,g}=\max \{O(e)\}_{e \in \bm{E}_{f, g}}$ & Threshold of group $(f, g)$ \\ \hline
$(\bm{E}_{f,g}, \tau_{f,g})^*$ & Ciphertext of $(\bm{E}_{f,g}, \tau_{f,g})$ \\ \hline
\end{tabular}

We say $EQ(Ercd)=1$ when $Ercd$ matches $EQ$.
Thus, the search result $EDB(EQ) = \{Ercd_{id} | EQ(Ercd_{id})=1\}$.
\end{table}
In this section, we give formal definitions for the search, access, and size patterns, as well as for forward and backward privacy.
Before that, in Table~\ref{tbl:notation}, we define the notations used throughout this article.

\begin{mydef}[\textbf{Search Pattern}]
Given a sequence of $q$ queries $\bm{Q}=(Q_1, \ldots, Q_q)$, the search pattern of $\bm{Q}$ represents the correlation between any two queries $Q_i$ and $Q_j$, \ie $\{Q_i\stackrel{?}{=}Q_j\}_{Q_i, Q_j \in \bm{Q}}$\footnote{$Q_i=Q_j$ only when $Q_i.type=Q_j.type$, $Q_i.f=Q_j.f$, $Q_i.op=Q_j.op$ and $Q_i.e=Q_j.e$}, where $1 \leq i, j \leq q$.
\end{mydef}
%

In previous works, access pattern is generally defined as the records matching each query \cite{Curtmola:2006:Searchable}, \ie the search result.
In fact, in leakage-based attacks, such as \cite{Zhang:2016:All,Islam:2012:Access,Cash:2015:leakage}, the attackers leverage the intersection between search results (explained in Section \ref{sec:attack}) to recover queries, rather than each single search result.
Therefore, in this work, we define the intersection between search results as access pattern.

\begin{mydef}[\textbf{Access Pattern}]
The access pattern of $\bm{Q}$ represents the intersection between any two search results, \ie $\{EDB(Q_i) \cap EDB(Q_j)\}_{Q_i, Q_j \in \bm{Q}}$.
\end{mydef}
\begin{mydef}[\textbf{Size Pattern}]
The size pattern of $\bm{Q}$ represents the number of records matching each query, \ie $\{|DB(Q_i)|\}_{Q_i\in \bm{Q}}$.
\end{mydef}
\begin{mydef}[\textbf{Forward Privacy}]
Let $Ercd^t$ be an encrypted record inserted or updated at time $t$, a dynamic SE scheme achieves forward privacy, if $EQ(Ercd^t)\stackrel{?}{=}0$ is always true for any query $EQ$ issued at time $t^*$, where $t^* < t$.
\end{mydef}
\begin{mydef}[\textbf{Backward Privacy}]
Let $Ercd^t$ be an encrypted record deleted at time $t$, a dynamic SE scheme achieves backward privacy, if $EQ(Ercd^t)\stackrel{?}{=}0$ is always true for any query $EQ$ issued at time $t'$, where $t < t'$.
\end{mydef}

%% file: sections/leakage.tex
\section{Leakage and Attacks}
\label{sec:leakage}
\subsection{Leakage Definition}
In \cite{Cash:2015:leakage}, Cash \etal define four different levels of leakage profiles for encrypted file collections according to the method of encrypting files and the data structure supporting encrypted search.
Yet, we cannot apply these definitions to databases directly, since the structure of a file is different from that of a record in the database.
In particular, a file is a collection of related words arranged in a semantic order and tagged with a set of keywords for searching; whereas, a record consists of a set of keywords with predefined attributes.
Moreover, a keyword may occur more than once in a file, and different keywords may have different occurrences; whereas, a keyword of an attribute generally occurs only once in a record.
Inspired by the leakage levels defined in \cite{Cash:2015:leakage}, in this section, we provide our own layer-based leakage definition for encrypted databases.
Specifically, we use the terminology \emph{leakage} to refer to the information the CSP can learn about the data directly from the encrypted database and the information about the results and queries when users are accessing the database.

The simplest type of SE scheme for databases is encrypting both the records and queries with Property-Preserving Encryption (PPE), such as the DETerministic (DET).
In DET-based schemes, the same data has the same ciphertext once encrypted.
In this type of SE schemes, the CSP can check whether each record matches the query efficiently by just comparing the corresponding ciphertext; however, these solutions result in information leakage.
Specifically, in DET-based schemes, such as CryptDB \cite{Popa:2011:Cryptdb} (where the records are protected only with the PPE layer), DBMask \cite{Sarfraz:2015:DBMask}, and Cipherbase \cite{Arasu:CIDR13:Cipherbase}, before executing any query, the CSP can learn the data distribution, \ie the number of distinct elements and the occurrence of each element, directly from the ciphertext of the database.
Formally, we say the data distribution of $DB$ is leaked if $e^*$ and $e$ have the same occurrence, \ie $O(e)=O(e^*)$, for each $e \in U$.
We define this leakage profile set as $\mathcal{L}_3$:
\begin{itemize}
  \item $\mathcal{L}_3=\{O(e)\}_{e \in U}$.
\end{itemize}

The second type of SE for databases encrypts the data with semantically secure primitives, but still encrypts the queries with DET encryption.
By doing so, the data distribution is protected, and the CSP can still search the encrypted database efficiently by repeating the randomisation over the DET query and then comparing it with the randomised data, as done in \cite{Hahn:2014:Searchable}, Arx \cite{Poddar:arx:eprint16}, and most of the Public-key Encryption with Keyword Search (PEKS) systems, such as \cite{BonehCOP04} and BlindSeer \cite{Fisch:SP15:BlindSeer}.
However, after executing a query, the CSP could still learn the access and size patterns.
Moreover, due to the DET encryption for queries, the search pattern is also leaked.
Given a sequence of $q$ queries $\textbf{Q}=(Q_1, \ldots, Q_q)$, we define the leakage profile as:
\begin{itemize}
  \item $\mathcal{L}_2=\{|DB(Q_i)|, \{EDB(Q_i)\cap EDB(Q_j), Q_i\stackrel{?}{=}Q_j\}_{Q_j \in \bm{Q}}\}_{Q_i \in \textbf{Q}}$
\end{itemize}
Note that after executing queries, PPE-based databases also leak the profiles included in $\mathcal{L}_2$.

A more secure SE solution leverages Oblivious RAM (ORAM)  \cite{Goldreich:1996:SPS,Stefanov:2013:PathORAM} or combines Homomorphic Encryption (HE) \cite{Paillier:1999:Public,Gentry:2009:FHE} with oblivious data retrieval to hide the search and access patterns.
For instance, the HE-based $PPQED_a$ proposed by Samanthula \etal \cite{Samanthula:2014:Privacy} and the ORAM-based SisoSPIR given by Ishai \etal \cite{Ishai:2016:Private} hide both the search and access patterns.
Unfortunately, in both schemes, the CSP can still learn how many records are returned to the user after executing a query, \ie \emph{the communication volume}.
According to \cite{Kellaris:ccs16:Generic}, the HE-based and ORAM-based SE schemes have fixed communication overhead between the CSP and users.
Specifically, the length of the message sent from the CSP to the user as the result of query execution is proportional to the number of records matching the query.
Based on this observation, the CSP can still infer the size pattern.
Thus, the HE-based and ORAM-based SE schemes are vulnerable to size pattern-based attacks, \eg count attack \cite{Cash:2015:leakage}.
The profile leaked in HE-based and ORAM-based SE schemes can be summarised below:
\begin{itemize}
  \item $\mathcal{L}_1=\{|DB(Q_i)|\}_{Q_i \in \bm{Q}}$.
\end{itemize}


%% file: sections/attacks.tex
\begin{table}
\scriptsize
  \centering
   \caption{Summary of leakage profiles and attacks against encrypted databases}
   \label{tbl:summary}
  \begin{tabular}{|c|l|l|}
     \hline
    \textbf{Leakage}  & \textbf{Schemes} & \textbf{Attacks} \\ \hline
    $\mathcal{L}_3$
    &
\begin{tabular}[c]{@{}l@{}}
    CryptDB \cite{Popa:2011:Cryptdb} \\
    DBMask \cite{Sarfraz:2015:DBMask} \\
    Cipherbase \cite{Arasu:CIDR13:Cipherbase}\\
    Monomi \cite{Tu:PVLDB13:monomi} \\
    Seabed \cite{Papadimitriou:usenix2016:Seabed} \\
\end{tabular}

& \begin{tabular}[c]{@{}l@{}}
    Frequency analysis attack \\
    IKK attack \\
    Count attack \\
    Record-injection attack \\
\end{tabular} \\ \hline

$\mathcal{L}_2$
&
\begin{tabular}[c]{@{}l@{}}
     Asghar \etal \cite{Asghar:2013:CCSW} \\
     Blind Seer \cite{Fisch:SP15:BlindSeer,Pappas:BlindSeer:SP14} \\
     Arx \cite{Poddar:arx:eprint16} \\
     PPQED \cite{Samanthula:2014:Privacy}
\end{tabular}

& \begin{tabular}[c]{@{}l@{}}
    IKK attack \\
    Count attack \\
    Record-injection attack \\
\end{tabular} \\ \hline

$\mathcal{L}_1$
&
\begin{tabular}[c]{@{}l@{}}
    $PPQED_a$ \cite{Samanthula:2014:Privacy} \\
    SisoSPIR \cite{Ishai:2016:Private}
\end{tabular}

&
    Count attack
 \\ \hline 


\end{tabular}
\end{table}
\subsection{Attacks against SE Solutions}
\label{sec:attack}
In recent years, leakage-based attacks against SE schemes have been investigated in the literature.
Table \ref{tbl:summary} summarises the existing SE solutions for relational databases and the attacks applicable to them.
In the following, we illustrate how the existing leakage-based attacks could recover the data and queries.
Specifically, for each attack, we analyse its leveraged leakage, required knowledge, process, and consequences.

\subsubsection{Frequency Analysis Attack}
In \cite{Naveed:2015:Inference}, Naveed \etal describe an attack on PPE-based SE schemes, where the CSP could recover encrypted records by analysing the leaked frequency information, \ie data distribution.
To succeed in this attack, in addition to the encrypted database, the CSP also requires some auxiliary information, such as the application background, publicly available statistics, and prior versions of the targeted database.
In PPE-based SE schemes, the frequency information of an encrypted database is equal to that of the database in plaintext.
By comparing the leaked frequency information with the obtained statistics relevant to the application, the CSP could recover the encrypted data elements stored in encrypted databases.
In \cite{Naveed:2015:Inference}, Naveed \etal recovered more than $60\%$ of records when evaluating this attack with real electronic medical records using CryptDB.
We stress that this attack does not require any queries or interaction with users.
The encrypted databases with $\mathcal{L}_3$ leakage profile, \ie PPE-based databases, such as CryptDB and DBMask, are vulnerable to this attack.

\subsubsection{IKK Attack}
IKK attack proposed by Islam \etal \cite{Islam:2012:Access} is the first attack exploiting the access pattern leakage.
The goal of the IKK attack is to recover encrypted queries in encrypted file collection systems, \ie recover the plaintext of searched keywords.
Note that this attack can also be used to recover queries in encrypted databases since it does not leverage the leakage specific to file collections.
In this attack, the CSP needs to know possible keywords in the dataset and the expected probability of any two keywords appearing in a file (\ie co-occurrence probability).
Formally, the CSP guesses $m$ potential keywords and builds an $m\times m$ matrix $\tilde{C}$ whose element is the co-occurrence probability of each keyword pair.
The CSP mounts the IKK attack by observing the access pattern revealed by the encrypted queries.
Specifically, by checking if any two queries match the same files or not, the number of files containing any two searched keywords (\ie the co-occurrence rate) can be reconstructed.
Assume the CSP observes $n$ queries.
It constructs an $n \times n$ matrix $C$ with their co-occurrence rates.
By using the simulated annealing technique \cite{KirkpatrickGV83}, the CSP can find the best match between $\tilde{C}$ and $C$ and map the encrypted keywords to the guesses.
In \cite{Islam:2012:Access}, Islam \etal mounted the IKK attack over the Enron email dataset \cite{eronemail:2017} and recovered $80\%$ of the queries with certain vocabulary sizes.
The encrypted relational databases with leakage profile $\mathcal{L}_2$ or $\mathcal{L}_1$, such as Arx \cite{Poddar:arx:eprint16}, Blind Seer \cite{Pappas:BlindSeer:SP14}, and PPQED \cite{Samanthula:2014:Privacy}, are also vulnerable to the IKK attack.


\subsubsection{File-injection and Record-injection Attack}
The file-injection attack \cite{Zhang:2016:All} is an active attack mounted on encrypted file collections, which is also named as \emph{chosen-document attack} in \cite{Cash:2015:leakage}.
The file-injection attack attempts to recover encrypted queries by exploiting access pattern in encrypted file storage.
More recently, Abdelraheem \etal \cite{Abdelraheem:eprint17:record} extended this attack to encrypted databases and defined it as \emph{record-injection attack}.
Compared with the IKK and count attack (will be discussed in Section \ref{subsec:conut}), much less auxiliary knowledge is required: the CSP only needs to know the keywords universe of the system.
In \cite{Zhang:2016:All}, Zhang \etal presented the first concrete file-injection attack and showed that the encrypted queries can be revealed with a small set of injected files.
Specifically, in this attack, the CSP (acting as an active attacker) sends files composed of the keywords of its choice, such as emails, to users who then encrypt and upload them to the CSP, which are called \emph{injected files}.
If no other files are uploaded simultaneously, the CSP can easily know the storage location of each injected file.
Moreover, the CSP can check which injected files match the subsequent queries.
Given enough injected files with different keyword combinations, the CSP could recover the keyword included in a query by checking the search result.
The encrypted databases with $\mathcal{L}_2$ or $\mathcal{L}_3$ leakage profiles are vulnerable to this attack.
Although some works \cite{BostMO17,ZuoSLSP18,ChamaniPPJ18,AmjadKM19} ensure both forward and backward privacy, they are still vulnerable to the file-injection attack due to the leakage of access pattern.
That is, after searching, the attacker could still learn the intersections between previous insert queries and the search result of current queries.

\subsubsection{Count and Relational-count Attack}
\label{subsec:conut}
The count attack is proposed by Cash \etal in \cite{Cash:2015:leakage} to recover encrypted queries in file storage systems based on the access and size patterns leakage.
In \cite{Abdelraheem:2017:Seachable}, Abdelraheem \etal have applied this attack to databases and named it a \emph{relational-count attack}.
As in the IKK attack scenario, the CSP is also assumed to know an $m \times m$ matrix $\tilde{C}$, where its entry $\tilde{C}[w_i, w_j]$ holds the co-occurrence rate of keyword $w_i$ and $w_j$ in the targeted dataset.
In order to improve the attack efficiency and accuracy, the CSP is assumed to know, for each keyword $w$, the number of matching files $count(w)$ in the targeted dataset.
The CSP mounts the count attack by counting the number of files matching each encrypted query.
For an encrypted query, if the number of its matching files is unique and equals to a known $count(w)$, the searched keyword must be $w$.
However, if the result size of a query $EQ$ is not unique, all the keywords with $count(w)=|EDB(EQ)|$ could be the candidates.
Recall that the CSP can construct another matrix $C$ that represents the observed co-occurrence rate between any two queries based on the leakage of access pattern.
By comparing $C$ with $\tilde{C}$, the candidates for the queries with non-unique result sizes can be reduced.
With enough recovered queries, it is possible to determine the keyword of $EQ$.
In \cite{Cash:2015:leakage}, Cash \etal tested the count attack against Enron email dataset and successfully recovered almost all the queries.
The SE solutions for databases with leakage profiles above $\mathcal{L}_1$ are vulnerable to this attack.

\subsubsection{Reconstruction Attack}
In ORAM-based systems, such as SisoSPIR proposed by Ishai \etal \cite{Ishai:2016:Private}, the size and access patterns are concealed.
Unfortunately, Kellaris \etal \cite{Kellaris:ccs16:Generic} observe that the ORAM-based systems have fixed communication overhead between the CSP and users, where the length of the message sent from the CSP to the user as the result of a query is proportional to the number of records matching the query.
That is, for a query $Q$, the size of the communication sent from the CSP to the user is $\alpha |DB(Q)|+ \beta$, where $\alpha$ and $\beta$ are two constants.
In theory, by giving two (query, result) pairs, the CSP can derive $\alpha$ and $\beta$, and then infer the result sizes of other queries.
In \cite{Kellaris:ccs16:Generic}, Kellaris \etal present the \emph{reconstruction attack} that exploits the leakage of communication volume, and could reconstruct the attribute names in encrypted databases supporting range queries.
In this attack, the CSP only needs to know the underlying query distribution prior to the attack.
Their experiment illustrated that after a certain number of queries, all the attributes can be recovered in a few seconds.
Since we focus on equality queries in this work, we do not give the attack details here.
Nonetheless, after recovering the size pattern for each query, the CSP could also mount the count attack on equality queries.
The SE schemes with $\mathcal{L}_1$ leakage profile are vulnerable to this attack.


%% file: sections/overview.tex
\section{Overview of \framework}
\label{sec:overview}
In this work, we propose \framework, a multi-cloud based dynamic SSE scheme for databases that can resist the aforementioned leakage-based attacks. 
Specifically, our scheme not only hides the frequency information of the database, but also protects the size, search, and access patterns.
Moreover, it ensures both forward and backward privacy when involving insert and delete queries.
Comparing with the existing SE solutions, \framework has the smallest leakage.
In this section, we define the system and threat model, and illustrate the techniques used in \framework at high-level.

\subsection{System Model}
\label{subsec:sm}
\begin{figure}
  \centering
  \includegraphics[width=.5\textwidth]{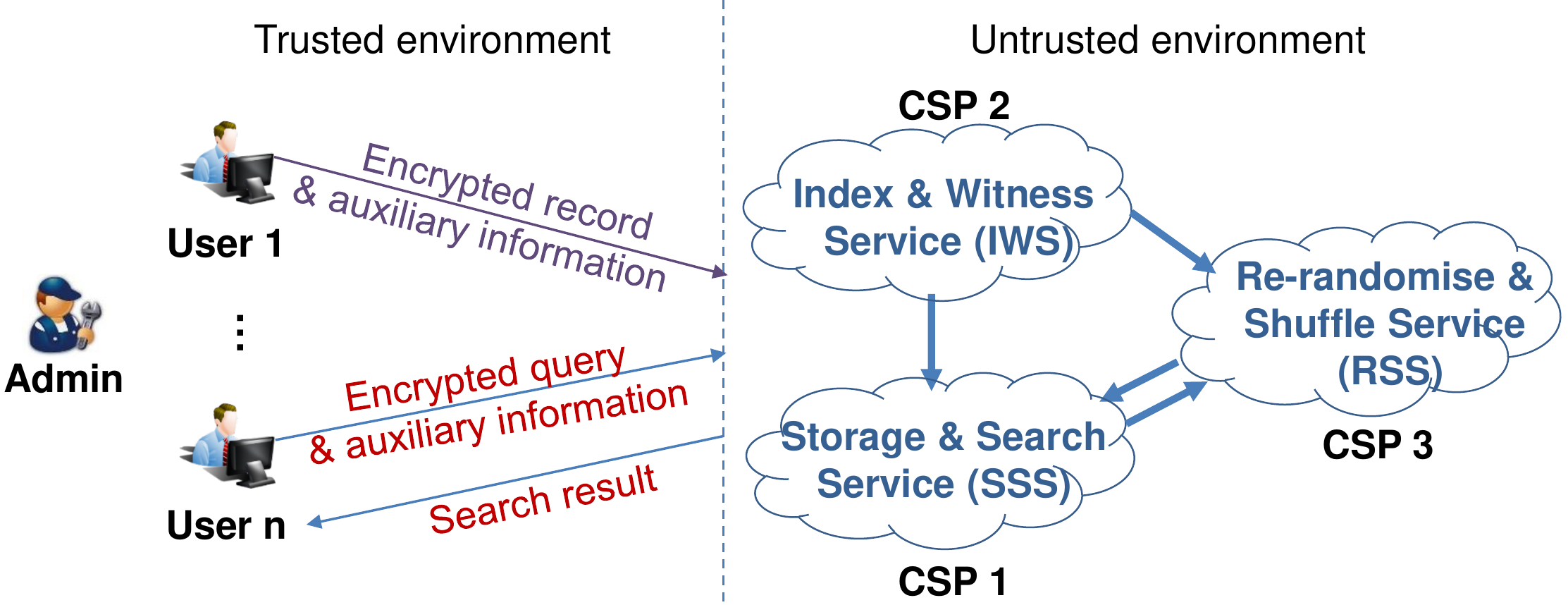}\\
  \caption{An overview of \framework:
  Users can upload records and issue queries.
  The SSS, IWS, and RSS represent independent CSPs.
  The SSS stores encrypted records and executes queries.
  The IWS stores index and auxiliary information, and provides witnesses to the SSS for performing encrypted search.
  After executing each query, the SSS sends searched records to the RSS for shuffling and re-randomising to protect patterns privacy.}
  \label{Fig:arch}
\end{figure}
In the following, we define our system model to describe the entities involved in \framework, as shown in Fig.~\ref{Fig:arch}:

 \begin{itemize}

   \item \textbf{Admin}: An admin is responsible for the setup and maintenance of databases, user management as well as specification and deployment of access control policies.

   \item \textbf{User}: A user can issue insert, select, delete, and update queries to read and write the database according to the deployed access control policies.
       \framework allows multiple users to read and write the database.

   \item \textbf{Storage and Search Service (SSS)}:
   It provides encrypted data storage, executes encrypted queries, and returns matching records in an encrypted manner.

   \item \textbf{Index and Witness Service (IWS)}:
    It stores the index and auxiliary information, and provides witnesses to the SSS for retrieving data.
    The IWS has no access to the encrypted data.

   \item \textbf{Re-randomise and Shuffle Service (RSS)}:
    After executing each query, it re-randomises and shuffles searched records to achieve the privacy of access pattern.
    The RSS does not store any data.

 \end{itemize}
    Each of the SSS, IWS, and RSS is deployed on the infrastructure managed by CSPs that are in conflict of interest. 
    According to the latest report given by RightScale \cite{RightScale:2016:report}, organisations are using more than three public CSPs on average, which means the schemes based on multi-cloud are feasible for most organisations. 
    The CSPs have to ensure that there is a two-way communication between any two of them, but our model assumes there is no collusion between the CSPs.

\subsection{Threat Model}
\label{subsec:tm}
We assume the admin is fully trusted.
All the users are only assumed to securely store their keys and the data.

The CSPs hosting the SSS, IWS, and RSS are modelled as honest-but-curious.
More specifically, they honestly perform the operations requested by users according to the designated protocol specification.
However, as mentioned in the above leakage-based attacks, they are curious to gain knowledge of records and queries by 1) analysing the outsourced data, 2) analysing the information leaked when executing queries, 3) and injecting malicious records.
As far as we know, \framework is the first SE scheme that considers active CSPs that could inject malicious records. 
Moreover, as assumed in \cite{Samanthula:2014:Privacy,Hoang:2016:practical,Stefanov:CCS2013:Multi-cloud}, we also assume the CSPs do not collude.
In other words, we assume an attacker could only compromise one CSP. 
In practice, any three cloud providers in conflict of interest, such as Amazon S3, Google Drive, and Microsoft Azure, could be considered since they may be less likely to collude in an attempt to gain information from their customers. 

We assume there are mechanisms in place for ensuring data integrity and availability of the system.

\subsection{Approach Overview}
\label{subsec:appo}
\framework aims at hiding the search, access, and size patterns.
\framework also achieves both backward and forward privacy.
We now give an overview of our approach.

To protect the search pattern, \framework XORs the query with a nonce, making identical queries look different once encrypted (\ie the encrypted query is semantically secure).
However, the CSP may still infer the search pattern by looking at the access pattern.
Specifically, the CSP can infer that two queries are equivalent if the same records are returned.
To address this issue, after executing each query, we shuffle the locations of the searched records.
Moreover, we re-randomise their ciphertexts, making them untraceable.
In this way, even if a query equivalent to the previous one is executed, the CSP will see a new set of records being searched and returned, and cannot easily infer the search and access pattern.

Another form of leakage is the size pattern, where the CSP can learn the number of records returned after performing a query, even after shuffling and re-randomisation.
Moreover, the CSP can guess the search pattern from the size pattern.
Specifically, the queries matching different numbers of records must be different, and the queries matching the same number of records could be equivalent.
To protect the size pattern, we introduce a number of dummy records that look exactly like the real ones and could match queries.
Consequently, the search result for each query will contain a number of dummy records making it difficult for the CSP to identify the actual number of real records returned by a query.

To break the link between size and search pattern, our strategy is to ensure all queries always match the same number of records, and the concrete method is to pad all the data elements in each field into the same occurrence with dummy records.
By doing so, the size pattern is also protected from the communication volume since there is no fixed relationship between them.
However, a large number of dummy records might be required for the padding.
To reduce required dummy records and ensure \framework's performance, we virtually divide the distinct data elements into groups and only pad the elements in the same group into the same occurrence.
By doing so, the queries searching values in the same group will always match the same number of records.
Then, the CSP cannot infer their search pattern.
Here we clarify that the search pattern is not fully protected in \framework.
Specifically, the CSP can still tell the queries are different if their search results are in different groups.

\framework also achieves forward and backward privacy.
Our strategy is to blind records also with nonces and re-randomise them using fresh nonces after executing each query.
Only queries that include the current nonce could match records.
In this way, even if a malicious CSP tries to use previously executed queries with old nonces, they will not be able to match the records in the dataset, ensuring forward privacy.
Similarly, deleted records (with old nonces) will not match newly issued queries because they use different nonces.

The details and algorithms of our scheme will be discussed in the following section.

%% file: sections/details.tex
\section{Solution details}
\label{sec:MCDB-details}
\begin{table}[htp]
\scriptsize
\centering
\caption{Data representation in \framework}
\subtable[\emph{Staff}]
{
\label{Tbl:mcdb-staff}
\begin{tabular}{|c|c|}\hline
\textbf{Name} & \textbf{Age} \\\hline
Alice & 27    \\\hline
Anna  & 30    \\\hline
Bob   & 27    \\\hline
Bill  & 25    \\\hline
Bob   & 33    \\\hline
\end{tabular}
}
\subtable[GDB on the IWS]
{
\label{Tbl:mcdb-engid}
\begin{tabular}{|c|c|c|c|}\hline
 \textbf{GID}    & \textbf{IL}        &$(\bm{E}, \tau)^*$         \\ \hline
 $(1, g_{1})$    & $\{1, 2\}$         & $(\{Alice, Anna\}, 1)^*$  \\\hline
 $(1, g_{2})$    & $\{3, 4, 5, 6\}$   & $(\{Bob, Bill\}, 2)^*$    \\\hline
 $(2, g'_{1})$   & $\{1, 3, 4, 6\}$   & $(\{25, 27\}, 2)^*$       \\\hline
 $(2, g'_{2})$   & $\{2, 5\}$         & $(\{30, 33\}, 1)^*$       \\\hline
\end{tabular}
}
\subtable[NDB on the IWS]
{
\label{Tbl:mcdb-nonceI}
\begin{tabular}{|c|c|c|}\hline
\textbf{id} &$\bm{seed}$ & $\bm{nonce}$ \\ \hline
1 & $seed_{1}$  & $\bm{n_1}$ \\\hline
2 & $seed_{2}$  & $\bm{n_2}$ \\\hline
3 & $seed_{3}$  & $\bm{n_3}$ \\\hline
4 & $seed_{4}$  & $\bm{n_4}$ \\\hline
5 & $seed_{5}$  & $\bm{n_5}$ \\\hline
6 & $seed_{6}$  & $\bm{n_6}$ \\\hline
\end{tabular}
}
\subtable[EDB on the SSS]
{
\label{Tbl:mcdb-enstaff}
\begin{tabular}{|c|c|c|c|c|}\hline
\textbf{ID} & \textbf{1}  & \textbf{2} & \textbf{Tag} \\\hline
1  & $SE(Alice)$  & $SE(27)$  & $tag_1$  \\\hline
2  & $SE(Anna)$   & $SE(30)$  & $tag_2$  \\\hline
3  & $SE(Bob)$    & $SE(27)$  & $tag_3$  \\\hline
4  & $SE(Bill)$   & $SE(25)$  & $tag_4$  \\\hline
5  & $SE(Bob)$    & $SE(33)$  & $tag_5$  \\\hline
6  & $SE(Bill)$   & $SE(25)$  & $tag_6$   \\\hline
\end{tabular}
}
\label{Tbl:mcdb-store}

 \subref{Tbl:mcdb-staff} A sample \emph{Staff} table.
 \subref{Tbl:mcdb-engid} GDB, the group information, is stored on the IWS.
 \subref{Tbl:mcdb-nonceI} NDB contains the seeds used to generate nonces.
 It might contain the nonces directly.
 NDB is also stored on the IWS.
 \subref{Tbl:mcdb-enstaff} EDB, the encrypted \emph{Staff} table, is stored on the SSS.
 Each encrypted data element $SE(e_f)=Enc_{s_1}(e_f)\oplus n_f$.
 Each record has a tag, enabling users to distinguish dummy and real records.
 In this example, the last record in Table~\subref{Tbl:mcdb-enstaff} is dummy.
 The RSS does not store any data.
\end{table}
\begin{algorithm}[tp]
\scriptsize
\caption{$Setup(k, DB)$}
\label{alg:mcdb-setup}
\begin{algorithmic}[1]
\STATE \underline{Admin: $KeyGen$}
\STATE $s_1, s_2 \lt \{0, 1\}^k$
~\\~
\STATE $GDB \lt \emptyset$, $EDB \lt \emptyset$, $NDB \lt \emptyset$
\STATE \underline{Admin: $GroupGen$} \label{code:setup-grpgen-be}
\FOR{each field $f$}
    \STATE Collect $U_f$ and $\{O(e)\}_{e \in U_f}$, and compute $\Psi_f=\{g \lt GE_{s_1}(e)\}_{e \in U_f}$
    \FOR{each $g \in \Psi_f$}
        \STATE $IL_{f, g} \lt \emptyset$, $\bm{E}_{f, g} \lt \{e\}_{e\in U_f \& GE_{s_1}(e)=g}$  \label{code:setup-grpgen-e}
        \STATE $\tau_{f, g} \lt \max\{|O(e)|\}_{e \in \bm{E}_{f, g}}$ \label{code:setup-grpgen-t}
        \STATE $(\bm{E}_{f, g}, \tau)^* \lt ENC_{s_1}(\bm{E}_{f, g}, \tau)$
            \label{code:setup-gdb}
        \STATE $GDB(f, g) \lt (IL_{f, g}, (\bm{E}_{f, g}, \tau_{f, g})^*)$
    \ENDFOR

\ENDFOR \label{code:setup-grpgen-end}

~\\~
\STATE \underline{Admin: DummyGen}\label{code:setup-dummygen-be}
\FOR{each field $f$}
    \STATE $\Sigma_f \lt \Sigma_{g \in \Psi_f}\Sigma_{e \in \bm{E}_{f, g}} (\tau_{f, g} - O(e))$
\ENDFOR
\STATE $\Sigma_{max} \lt \max\{\Sigma_1, \ldots, \Sigma_F\}$
\STATE Add $\Sigma_{max}$ dummy records with values $(NULL, \ldots, NULL)$ into $DB$
\FOR {each field $f$}
    \FOR {each $e \in U_f$}
        \STATE Assign $e$ to $\tau_{f, GE_{s_1}(e)}-O(e)$ dummy records in field $f$
    \ENDFOR
\ENDFOR
\STATE Mark real and dummy records with $flag=1$ and $flag=0$, respectively
\STATE Shuffle $DB$ \label{code:setup-dummygen-end}

~\\~
\STATE \underline{Admin: DBEnc} \label{code:setup-DBenc-be}
\STATE $id \lt 0$
\FOR{each $rcd \in DB$}
    \STATE $(Ercd, seed, \bm{n}, Grcd)\lt RcdEnc(rcd, flag)$
    \STATE $EDB(id) \lt Ercd$, $NDB(id) \lt (seed, \bm{n})$
    \FOR{each $g_f \in Grcd$}
        \STATE $IL_{f, g} \lt IL_{f, g} \cup id$  \label{code:setup-DBenc-IL}
    \ENDFOR
    \STATE $id++$  \label{code:setup-DBenc-end}
\ENDFOR
\end{algorithmic}
\end{algorithm}

In this section, we give the details for setting up, searching, and updating the database.

\subsection{Setup}
\label{subsec:boot}
The system is set up by the admin by generating the secret keys $s_1$ and $s_2$ based on the security parameter $k$.
$s_1$ is only known to users and is used to protect queries and records from CSPs.
$s_2$ is generated for saving storage, and it is known to both the user and IWS and is used to generate nonces for record and query encryption.
The admin also defines the cryptographic primitives used in \framework.

We assume the initial database $DB$ is not empty.
The admin bootstraps $DB$ with Algorithm \ref{alg:mcdb-setup}, $Setup(k, DB) \rt (EDB, GDB, NDB)$.
Roughly speaking, the admin divides the records into groups (Lines \ref{code:setup-grpgen-be}-\ref{code:setup-grpgen-end}), pads the elements in the same group into the same occurrence by generating dummy records (Lines \ref{code:setup-dummygen-be}-\ref{code:setup-dummygen-end}), and encrypts each record (Lines \ref{code:setup-DBenc-be}-\ref{code:setup-DBenc-end}).
The details of each operation are given below.

\paragraph{Group Generation}
As mentioned, inserting dummy records is necessary to protect the size and search patterns, and grouping the data aims at reducing the number of required dummy records.

Indeed, dividing the data into groups could also reduce the number of records to be searched.
Only padding the data in the same group into the same occurrence could result in the leakage of group information.
Particularly, the SSS can learn if records and queries are in the same group from the size pattern.
Considering the group information will be inevitably leaked after executing queries, \framework allows the SSS to know the group information in advance, and only search a group of records for each query rather than the whole database.
By doing so, the query can be processed more efficiently without leaking additional information.
Yet, in this case, the SSS needs to know which group of records should be searched for each query.
Considering the SSS only gets encrypted records and queries, the group should be determined by the admin and users.
To avoid putting heavy storage overhead on users, \framework divides data into groups with a Pseudo-Random Function (PRF) $GE: \{0, 1\}^* \times \{0, 1\}^k \rightarrow \{0, 1\}^*$.
The elements in field $f$ ($1\leq f \leq F$) with the same $g \leftarrow GE_{s_1}(e)$ value are in the same group, and $(f, g)$ is the group identifier.
In this way, the admin and users can easily know the group identifiers of records and queries just by performing $GE$.

The implementation of $GE$ function affects the security level of the search pattern.
Let $\lambda$ stand for the number of distinct elements contained in a group.
Since the elements in the same group will have the same occurrence, the queries involving those elements (defined as \emph{the queries in the same group}) will match the same number of records.
Then, the adversary cannot tell their search patterns from their size patterns.
Formally, for any two queries matching the same number of records, the probability they involve the same keyword is $\frac{1}{\lambda}$.
Thus, $\lambda$ also represents the security level of the search pattern.
Given $\lambda$, the implementation of $GE$ should ensure each group contains at least $\lambda$ distinct elements.
For instance, the admin could generate the group identifier of $e$ by computing $LSB_b(H_{s_1}(e))$, where $LSB_b$ gets the least significant $b$ bits of its input.
To ensure each group contains at least $\lambda$ distinct elements, $b$ can be smaller.

The details of grouping $DB$ are shown in Lines \ref{code:setup-grpgen-be}-\ref{code:setup-grpgen-end} of Algorithm \ref{alg:mcdb-setup}.
Formally, we define group $(f, g)$ as $(IL_{f, g}, \bm{E}_{f, g}, \tau_{f, g})$, where $IL_{f, g}$ stores the identifiers of the records in this group (Line \ref{code:setup-DBenc-IL}), $\bm{E}_{f, g}$ is the set of distinct elements in this group (Line \ref{code:setup-grpgen-e}), and $\tau_{f, g}=\max\{O(e) | e \in \bm{E}_{f, g}\}$ is the occurrence threshold for padding (Line \ref{code:setup-grpgen-t}).
Since the group information will be stored in the CSP, $(\bm{E}_{f, g}, \tau_{f, g})$ is encrypted into $(\bm{E}_{f, g}, \tau_{f, g})^*$ with $s_1$ and a semantically secure symmetric encryption primitive $ENC: \{0, 1\}^* \times \{0, 1\}^k \rt \{0, 1\}^*$.
$(\bm{E}_{f, g}, \tau_{f, g})^*$ is necessary for insert queries (The details are given in Section~\ref{subsec:update}).

Note that if the initial database is empty, the admin can pre-define a possible $U_f$ for each field and group its elements in the same way.
In this case, $IL=\emptyset$ and $\tau=0$ for each group after the bootstrapping.

\paragraph{Dummy Records Generation}
Once the groups are determined, the next step is to generate dummy records.
The details for generating dummy records are given in Lines \ref{code:setup-dummygen-be}-\ref{code:setup-dummygen-end}, Algorithm \ref{alg:mcdb-setup}.
Specifically, the admin first needs to know how many dummy records are required for the padding.
Since the admin will pad the occurrence of each element in $\bm{E}_{f, g}$ into $\tau_{f, g}$, $\tau_{f, g}-O(e)$ dummy records are required for each $e \in \bm{E}_{f, g}$.

Assume there are $M$ groups in field $f$, then $\Sigma_f=\sum_{i=1}^{M}\sum_{e \in \bm{E}_{f, g^i}}(\tau_{f, g^i}-O(e))$ dummy records are required totally for padding field $f$.
For the database with multiple fields, different fields might require different numbers of dummy records.
Assume $\Sigma_{max}=\max\{\Sigma_1, \ldots, \Sigma_F\}$.
To ensure all fields can be padded properly, $\Sigma_{max}$ dummy records are required.
Whereas, $\Delta_f=\Sigma_{max} - \Sigma_f$ dummy records will be redundant for field $f$.
The admin assigns them a meaningless string, such as `NULL', in field $f$.
After encryption, `NULL' will be indistinguishable from other elements.
Thus, the CSP cannot distinguish between real and dummy records.
Note that users and the admin can search the records with `NULL'.

In this work, we do not consider the query with conjunctive predicates, so we do not consider to pad the element pairs also into the same occurrence.

After padding, each record $rcd$ is appended with a $flag$ to mark if it is real or dummy.
Specifically, $flag=1$ when $rcd$ is real, otherwise $flag=0$.
The admin also shuffles the dummy and real records.

\paragraph{Record Encryption}
\begin{algorithm}[tp]
\scriptsize
\caption{$RcdEnc(rcd, flag)$}
\label{alg:mcdb-enc}
\begin{algorithmic}[1]
   \STATE $seed \stackrel{\$}{\leftarrow} \{0,1\}^{|seed|}$
   \STATE $\bm{n} \leftarrow \Gamma_{s_2} (seed)$, where $\bm{n}=\ldots \| n_f \| \ldots \| n_{F+1}$, $|n_f|=|e|$ and $|n_{F+1}|=|H|+|e|$ \label{code:mcdb-enc-seed}
    \FOR {each element $e_f \in rcd$}
        \STATE $g_f \lt GE_{s_1}(e_f)$  \label{code:mcdb-enc-gid}
        \STATE $e^*_f \lt  Enc_{s_1}(e_f) \oplus n_f$  \label{code:mcdb-enc-se}
    \ENDFOR
    \IF {$flag=1$}
        \STATE $S \stackrel{\$}{\leftarrow} \{0,1\}^{|e|}$
        \STATE $tag \leftarrow (H_{s_1}(S)||S )\oplus n_{F+1}$ \label{code:mcdb-enc-tag-re}
    \ELSE
        \STATE $tag \stackrel{\$}{\leftarrow} \{0,1\}^{|H|+|e|}$ \label{code:mcdb-enc-tag-du}
    \ENDIF
\RETURN $Ercd=(e^*_1, \ldots, e^*_F, tag)$, $(seed, \bm{n})$, and $Grcd=(g_1, \ldots, g_F)$
\end{algorithmic}
\end{algorithm}
The admin encrypts each record before uploading them to the SSS.
The details of record encryption are provided in Algorithm~\ref{alg:mcdb-enc}, $RcdEnc (rcd, flag) \rightarrow (Ercd, seed, \bm{n}, Grcd)$.

To ensure the dummy records could match queries, they are encrypted in the same way as real ones.
Specifically, first the admin generates a random string as a $seed$ for generating a longer nonce $\bm{n}$ with a Pseudo-Random Generator (PRG) $\Gamma: \{0, 1\}^{|seed|} \times \{0, 1\}^{k} \rightarrow \{0, 1\}^{*}$ (Line~\ref{code:mcdb-enc-seed}, Algorithm~\ref{alg:mcdb-enc}).
Second, the admin generates $g_f$ for each $e_f \in rcd$ by computing $GE_{s_1}(e_f)$ (Line \ref{code:mcdb-enc-gid}).
Moreover, $e_f$ is encrypted by computing $SE(e_f): e^*_f \leftarrow Enc_{s_1}(e_f) \oplus n_f$ (Line~\ref{code:mcdb-enc-se}), where $Enc:\{0,1\}^* \times \{0,1\}^k \rightarrow \{0,1\}^* $ is a deterministic symmetric encryption primitive, such as AES-ECB.
Using $n_f$, on the one hand, ensures the semantically secure of $e^*_f$.
On the other hand, it ensures the forward and backward privacy of \framework (as explained in Section \ref{subsec:shuffle}).
$e^*_f$ will be used for encrypted search and data retrieval.

The dummy records are meaningless items, and the user does not need to decrypt returned dummy records.
Thus, we need a way to filter dummy records for users.
Considering the CSPs are untrusted, we cannot mark the real and dummy records in cleartext.
Instead, we use a keyed hash value to achieve that.
Specifically, as shown in Lines~\ref{code:mcdb-enc-tag-re} and \ref{code:mcdb-enc-tag-du}, a tag $tag$ is generated using a keyed hash function $H: \{0,1\}^* \times \{0,1\}^k \rightarrow \{0,1\}^*$ and the secret key $s_1$ if the record is real, otherwise $tag$ is a random bit string.
With the secret key $s_1$, the dummy records can be efficiently filtered out by users before decrypting the search result by checking if:
\begin{align}\label{eq:check1}\notag
  tag_{l} \stackrel{?}{=}& H_{s_1}(tag_{r}), \text{ where } tag_{l}||tag_{r} = tag \oplus n_{F+1}
\end{align}

Once all the real and dummy records are encrypted, the admin uploads the auxiliary information, \ie the set of group information $GDB$ and the set of nonce information $NDB$, to the IWS, and uploads encrypted records $EDB$ to the SSS.
$GDB$ contains $(IL, (\bm{E}, \tau)^*)$ for each group. 
$NDB$ contains a $(seed, \bm{n})$ pair for each record stored in $EDB$. 
To reduce the storage overhead on the IWS, $NDB$ could also just store the seed and recover $\bm{n}$ by computing $\Gamma_{s_2}(seed)$ when required.
Whereas, saving the $(seed, \bm{n})$ pairs reduces the computation overhead on the IWS.
In the rest of this article, we assume NDB contains $(seed, \bm{n})$ pairs.
$EDB$ contains the encrypted record $Ercd$. 
Note that to ensure the correctness of the search functionality, it is necessary to store the encrypted records and their respective $(seed, \bm{n})$ pairs in the same order in $EDB$ and $NDB$ (the search operation is explained in Section \ref{subsec:mcdb-select}).
In Table~\ref{Tbl:mcdb-store}, we take the \emph{Staff} table (\ie Table \ref{Tbl:mcdb-staff}) as an example and show the details stored in $GDB$, $NDB$, and $EDB$ in Tables~\ref{Tbl:mcdb-engid},~\ref{Tbl:mcdb-nonceI}, and~\ref{Tbl:mcdb-enstaff}, respectively.


%


\subsection{Select Query}
\label{subsec:mcdb-select}
\begin{algorithm}[htp]
\scriptsize
\caption{Query$(Q)$}
\label{alg:macdb-search}
\begin{algorithmic}[1]
\STATE \underline{User: QueryEnc$(Q)$}  \label{code:mcdb-query-user-be}
    \STATE $g \leftarrow GE_{s_1}( Q.e )$ \label{code:mcdb-query-user-gid}
    \STATE $EQ.type \lt Q.type$, $EQ.f \lt Q.f$
    \STATE $\eta \stackrel{\$}{\lt} \{0,1\}^{|e|}$
    \STATE $EQ.e^*  \lt Enc_{s_1}(Q.e) \oplus \eta$  \label{code:mcdb-query-se}
    \STATE Send $EQ=(type, f, e^*)$ to the SSS
    \STATE Send $(EQ.f, \eta, g)$ to the IWS \label{code:mcdb-query-user-end}
%
~\\~

\STATE \underline{IWS: $NonceBlind(EQ.f, \eta, g)$} \label{code:mcdb-query-IWS-be}
    \STATE $EN \leftarrow \emptyset$
    \STATE $IL \leftarrow GDB(EQ.f, g)$ \COMMENT{If $(EQ.f, g) \notin GDB$, return $IL$ of the  closest group(s).}\label{code:mcdb-query-IL}
    \FOR {each $id \in IL$ }
        \STATE  $(seed, \bm{n}) \leftarrow NDB(id)$, where $\bm{n}= \ldots ||n_{EQ.f}|| \ldots$ and $|n_{EQ.f}|=|\eta|$
                \label{code:mcdb-query-get-nonce}
        \STATE  $w \lt H'(n_{EQ.f} \oplus \eta)$
        \STATE  $t \lt \eta \oplus seed$
        \STATE  $EN(id) \leftarrow (w, t)$ \label{code:mcdb-query-enc-nonce}
    \ENDFOR
    \STATE Send $IL=(id, \ldots)$ and the encrypted nonce set $EN=((w, t), \ldots)$ to the SSS   \label{code:mcdb-query-IWS-end}

~\\
\STATE \underline{SSS: $Search(EQ, EN, IL)$}   \label{code:mcdb-query-SSS-be}
    \STATE $SR \leftarrow \emptyset$
    \FOR {each $id \in IL$}
        \IF {$H'(EDB(id, EQ.f) \oplus EQ.e^*) = EN(id).w$ }  \label{code:mcdb-query-checke}
             \STATE $SR \leftarrow SR \cup (EDB(id), EN(id).t)$ \label{code:mcdb-query-sr}
        \ENDIF
    \ENDFOR
    \STATE Send the search result $SR$ to the user \label{code:mcdb-query-SSS-end}

~\\
\STATE \underline{User: RcdDec$(SR, \eta)$} \label{code:mcdb-query-dec-be}
\FOR{each $(Ercd, t) \in SR$}
    \STATE  $\bm{n} \leftarrow \Gamma_{s_2} (t \oplus \eta)$ \label{code:mcdb-query-dec-seed}
    \STATE  $(Enc_{s_1}(rcd), tag) \leftarrow Ercd \oplus \bm{n}$
    \STATE $tag_{l} || tag_{r} \lt tag$, where $|tag_r|=|e|$
    \IF{$tag_{l} = H_{s_1}(tag_{r})$} \label{code:mcdb-query-ducheck}
        \STATE  $rcd \leftarrow Enc^{-1}_{s_1} (Enc_{s_1}(rcd))$ \label{code:mcdb-query-dec}
    \ENDIF
\ENDFOR \label{code:mcdb-query-dec-end}
\end{algorithmic}
\end{algorithm}
In this work, we focus on the simple query which only has one single equality predicate.
The complex queries with multiple predicates can be performed by issuing multiple simple queries and combing their results on the user side.
To support range queries, the technique used in \cite{Asghar:Espoon:IRC13} can be adopted. 

For performing a select query, \framework requires the cooperation between the IWS and SSS.
The details of the steps performed by the user, IWS, and SSS are shown in Algorithm~\ref{alg:macdb-search}, $Query(Q) \rightarrow SR$, which consists of 4 components: $QueryEnc$, $NonceBlind$, $Search$, $RcdDec$.

\paragraph{\bm{$QueryEnc(Q)\rightarrow (EQ, \eta, g)$}}
First, the user encrypts the query $Q=(type, f, e)$ using $QueryEnc$ (Lines \ref{code:mcdb-query-user-be} - \ref{code:mcdb-query-user-end}, Algorithm \ref{alg:macdb-search}).
Specifically, to determine the group to be searched, the user first generates $g$ (Line \ref{code:mcdb-query-user-gid}).
We do not aim at protecting the query type and searched field from CSPs.
Thus, the user does not encrypt $Q.type$ and $Q.f$.
The interested keyword $Q.e$ is encrypted into $EQ.e^*$ by computing $Enc_{s_1}(Q.e)\oplus \eta$ (Line \ref{code:mcdb-query-se}).
The nonce $\eta$ ensures that $EQ.e^*$ is semantically secure.
Finally, the user sends $EQ=(type, f, e^*)$ to the SSS and sends $(EQ.f$, $\eta$, $g)$ to the IWS.

\paragraph{$\bm{NonceBlind(EQ.f, \eta, g)\rightarrow (IL, EN)}$}
Second, the IWS provides $IL$ and witnesses $EN$ of group $(EQ.f, g)$ to the SSS by running $NonceBlind$ (Line \ref{code:mcdb-query-IWS-be} - \ref{code:mcdb-query-IWS-end}).
Specifically, for each $id \in IL$, the IWS generates $EN(id)=(w, t)$ (Lines~\ref{code:mcdb-query-get-nonce}-\ref{code:mcdb-query-enc-nonce}),
where $w=H'(n_{EQ.f} \oplus \eta)$ will be used by the SSS to find the matching records, and $t= \eta \oplus seed$ will be used by the user to decrypt the result.
Here $H': \{0, 1\}^* \rt \{0, 1\}^k$ is a hash function.
Note that when $(EQ.f, g)$ is not contained in $GDB$, $IL$ of the \emph{closest} group(s) will be used, \ie the group in field $EQ.f$ whose identifier has the most common bits with $g$\footnote{This can be obtained by comparing the hamming weight of $g' \oplus g$ for all $(EQ.f, g') \in GDB$.}.

\paragraph{$\bm{Search(EQ, IL, EN)\rightarrow SR}$}
Third, the SSS traverses the records indexed by $IL$ and finds the records matching $EQ$ with the assistance of $EN$ (Lines \ref{code:mcdb-query-SSS-be} - \ref{code:mcdb-query-SSS-end}).
Specifically, for each record indexed by $IL$, the SSS checks if $H'(EDB(id, EQ.f) \oplus EQ.e^*) \stackrel{?}{=} EN(id).w$ (Line \ref{code:mcdb-query-checke}).
More specifically, the operation is:
\begin{equation}\notag
 H'(Enc_{s_1}(e_{EQ.f})\oplus n_{EQ.f}  \oplus Enc_{s_1}(Q.e) \oplus \eta ) \stackrel{?}{=} H'(n_{EQ.f} \oplus \eta)
\end{equation}
It is clear that only when $Q.e=e_{EQ.f}$ there is a match.
The SSS sends each matched record $EDB(id)$ and its corresponding $EN(id).t$ to the user as the search result $SR$, \ie $EDB(EQ)$.

\paragraph{\bm{$RcdDec(SR) \rightarrow rcds$}}
To decrypt an encrypted record $Ercd$, both the secret key $s_1$ and nonce $\bm{n}$ are required.
The nonce $\bm{n}$ can be recovered from the returned $t$.
Only the user issuing the query knows $\eta$ and is able to recover $\bm{n}$ by computing $\Gamma_{s_2}(t \oplus \eta)$ (Line \ref{code:mcdb-query-dec-seed}).
With $\bm{n}$, the user can check if each returned record is real or dummy (Line \ref{code:mcdb-query-ducheck}), and decrypt each real record by computing
$Enc^{-1}_{s_1}(Ercd \oplus \bm{n})$ (Line \ref{code:mcdb-query-dec}), where $Enc^{-1}$ is the inverse of $Enc$.

\subsection{Shuffling and Re-randomisation}
\label{subsec:shuffle}
\begin{algorithm}
\scriptsize
\caption{$Shuffle(IL, Ercds)$}
\label{alg:mcdb-shuffle}
\begin{algorithmic}[1]
\STATE \underline{IWS: $PreShuffle(IL)$} \label{code:mcdb-preshuffle-be}
    \STATE $IL' \leftarrow \pi (IL)$
    \STATE Shuffle the $(seed, \bm{n})$ pairs indexed by $IL$ based on $IL'$
    \STATE Update the indices of affected groups in $GDB$ \label{code:mcdb-shuffle-update-groups}
    \FOR {each $id \in IL'$} \label{code:mcdb-shuffle-begin}
        \STATE $seed \stackrel{\$}{\leftarrow} \{0, 1\}^{|seed|}$ \label{code:mcdb-seed'}
        \STATE $\bm{n'} \leftarrow  \Gamma_{s_2}(seed)$ \label{code:mcdb-gn'}
        \STATE $NN(id) \leftarrow NDB(id).\bm{n} \oplus \bm{n'}$ \label{code:mcdb-nn}
        \STATE $NDB(id) \leftarrow (seed, \bm{n'})$ \label{code:mcdb-n'}
    \ENDFOR
    \STATE Send $(IL', NN)$ to the RSS. \label{code:mcdb-preshuffle-end}

~\\~
\STATE{\underline{RSS: $Shuffle(Ercds, IL', NN)$}}
    \STATE Shuffle $Ercds$ based on $IL'$
    \FOR {each $id \in IL'$}
        \STATE $Ercds(id) \leftarrow Ercds(id) \oplus NN(id) $   \label{code:mcdb-shuffle-reenc}
    \ENDFOR
    \STATE Send $Ercds$ to the SSS.
\end{algorithmic}
\end{algorithm}

To protect the access pattern and ensure the forward and backward privacy, \framework shuffles and re-randomises searched records after executing each query, and this procedure is performed by the IWS and RSS.
The details are shown in Algorithm \ref{alg:mcdb-shuffle}, consisting of $PreShuffle$ and $Shuffle$.

\paragraph{$\bm{PreShuffle(IL) \rightarrow (IL', NN)}$}
In \framework, the searched records are re-randomised by renewing the nonces.
Recall that $SE$ encryption is semantically secure due to the nonce.
However, the IWS stores the nonces.
If the IWS has access to encrypted records, it could observe deterministically encrypted records by removing the nonces.
To void leakage, \framework does not allow the IWS to access any records and involves the RSS to shuffle and re-randomise the records.

Yet, the IWS still needs to shuffle $NDB$ and generate new nonces for the re-randomisation by executing $PreShuffle$.
Specifically, as shown in Algorithm \ref{alg:mcdb-shuffle}, Lines \ref{code:mcdb-preshuffle-be}-\ref{code:mcdb-preshuffle-end}, the IWS first shuffles the $id$s in $IL$ with a Pseudo-Random Permutation (PRP) $\pi$ and gets the re-ordered indices list $IL'$.
In our implementation, we leverage the modern version of the Fisher-Yates shuffle algorithm \cite{Knuth73}, where from the first $id$ to the last one, each $id$ in $IL$ is exchanged with a random $id$ storing behind it.
After that, the IWS shuffles $(seed, \bm{n})$ pairs based on $IL'$.
Note that the shuffling operation affects the list of indices of the groups in other fields.
Thus, the IWS also needs to update the index lists of other groups accordingly (Line \ref{code:mcdb-shuffle-update-groups}).
For re-randomising records, the IWS samples a new seed and generates a new nonce $\bm{n'}$.
To ensure the records will be blinded with the respective new nonces stored in $NDB$ after shuffling and re-randomising, the IWS generates $NN=(\bm{n}\oplus\bm{n'}, \ldots)$ for RSS (Line \ref{code:mcdb-nn}).
Afterwards, IWS updates the seed and nonce stored in $NDB(id)$ with the new values.
Finally, $(IL', NN)$ is sent to the RSS.


\paragraph{$\bm{Shuffle(Ercds, IL', NN) \rightarrow Ercds}$}
After searching, the SSS sends the searched records $Ercds$ to the RSS.
Given $IL'$ and $NN$, the RSS starts to shuffle and re-randomise $Ercds$.
Specifically, the RSS first shuffles $Ercds$ based on $IL'$, and then re-randomises each record by computing $Ercds(id) \oplus NN(id)$ (Line \ref{code:mcdb-shuffle-reenc}).
In details, the operation is:
\begin{equation}\notag
 (Enc_{s_1}(rcd_{id})\oplus \bm{n})  \oplus (\bm{n'} \oplus \bm{n}) = Enc_{s_1}(rcd_{id})\oplus \bm{n'}
\end{equation}
That is, $Ercds(id)$ is blinded with the latest nonce stored in $NDB(id)$.
Finally, the re-randomised and shuffled records $Ercds$ are sent back to the SSS.

By using a new set of seeds for the re-randomisation, \framework achieves both forward and backward privacy.
If the SSS tries to execute an old query individually, it will not be able to match any records without the new witness $w$, which can only be generated by the IWS with new nonces.
Similarly, the SSS cannot learn if deleted records match new queries.

\subsection{User Revocation}
\label{sec:revoke}
\framework supports flexible multi-user access in a way that the issued queries and search results of one user are protected from all the other entities. Moreover, revoking users do not require key regeneration and data re-encryption even when one of the CSPs colludes with revoked users.

As mentioned in Section \ref{subsec:mcdb-select}, for filtering dummy records and recovering returned real records, both $s_1$ and the nonce are required.
After shuffling, the nonce is only known to the IWS.
Thus, without the assistance of the IWS and SSS, the user is unable to recover records only with $s_1$.
Therefore, for user revocation, we just need to manage a revoked user list at the IWS as well as at the SSS.
Once a user is revoked, the admin informs the IWS and SSS to add this user into their revoked user lists.
When receiving a query, the IWS and the SSS will first check if the user has been revoked.
If yes, they will reject the query.
In case revoked users collude with either the SSS or IWS, they cannot get the search results, since such operation requires the cooperation of both the user issuing the query, IWS, and SSS.

\subsection{Database Updating}
\label{subsec:update}
\begin{algorithm}
\scriptsize
\caption{$Insert(rcd)$}
\label{alg:mcdb-insert}
\begin{algorithmic}[1]
\STATE \underline{$User(rcd)$}:
\STATE $(Ercd, seed, \bm{n}, Grcd) \lt RcdEnc(rcd, 1)$ \label{code:mcdb-insert-user-enc}
\STATE $INS_{IWS} \lt (seed, \bm{n}, Grcd)$, $INS_{SSS} \lt  Ercd$
\FOR {each $g_f \in Grcd$}   \label{code:mcdb-insert-user-be}
    \STATE $(\bm{E}_{f, g_f}, \tau_{f, g_f})^* \lt GDB(f, g_f)$\COMMENT{If $(f, g_f) \notin GDB$, $g_f \lt g_f'$, where $(f, g_f')$ is the closet group of $(f, g_f)$.}
    \STATE $(\bm{E}_{f, g_f}, \tau_{f, g_f}) \lt ENC^{-1}_{s_1}((\bm{E}_{f, g_f}, \tau_{f, g_f})^*)$ \label{code:mcdb-insert-user-ef}
\ENDFOR
\FOR{each $e_f \in rcd$}  \label{code:mcdb-insert-gendum-be}
    \IF{$e_f \in \bm{E}_{f, g_f}$}
        \STATE $\gamma_{f} \lt |\bm{E}_{f, g_f}| -1 $
    \ELSE
        \STATE $\gamma_{f} \lt \tau_{f, g_f} -1 $
    \ENDIF
\ENDFOR
\STATE $W=\max\{\gamma_{f}\}_{1 \leq f \leq F}$ \label{code:mcdb-insert-encdum-be}
\STATE Generate $W$ dummy records with values $(NULL, \ldots, NULL)$
\FOR{each $e_f \in rcd$}
    \IF{$e_f \in \bm{E}_{f, g_f}$}
        \STATE Assign $\bm{E}_{f, g_f} \setminus e_f$ to $\gamma_f$ dummy records in field $f$
        \STATE $\tau_{f, g_f} ++$
    \ELSE
        \STATE Assign $e_f$ to $\gamma_f$ dummy records in field $f$
        \STATE $\bm{E}_{f, g_f} \lt \bm{E}_{f, g_f} \cup e_f$
    \ENDIF
    \STATE $(\bm{E}_{f, g_f}, \tau_{f, g_f})^* \lt ENC_{s_1}(\bm{E}_{f, g_f}, \tau_{f, g_f})$  \label{code:mcdb-insert-gendum-end}
\ENDFOR
\FOR{each dummy record $rcd'$}
    \STATE $(Ercd, seed, \bm{n}, Grcd) \lt RcdEnc(rcd', 0)$ \label{code:mcdb-insert-enc-du}
    \STATE $INS_{IWS} \lt INS_{IWS} \cup (seed, \bm{n}, Grcd)$
    \STATE $INS_{SSS} \lt INS_{SSS} \cup Ercd$ \label{code:mcdb-insert-user-end}
\ENDFOR
\STATE Send $INS_{IWS}$ and $((\bm{E}_{f, g_f}, \tau_{f, g_f})^*)_{1 \leq f \leq F}$ to the IWS
\STATE Send $INS_{SSS}$ to the SSS

~\\

\STATE \underline{$SSS(INS_{SSS}$}):  \label{code:mcdb-insert-csp-be}
\STATE $IDs \lt \emptyset$
\FOR{each $Ercd \in INS_{SSS}$}
    \STATE $EDB(++id) \lt Ercd$  \label{code:mcdb-insert-csp-insert}
    \STATE $IDs \lt IDs \cup id$
\ENDFOR
\STATE Send $IDs$ to the IWS  \label{code:mcdb-insert-csp-end}

~\\

\STATE \underline{$IWS(INS_{IWS}, (\bm{E}_{f, g_f}, \tau_{f, g_f})^*)_{1 \leq f \leq F}, IDs)$}: \label{code:mcdb-insert-wss-be}
\FOR{each $(seed, \bm{n}, Grcd) \in INS_{IWS}$ and $id \in IDs$} \label{code:mcdb-insert-wss-seed}
    \STATE $NDB(id) \lt (seed, \bm{n})$
    \FOR{$f=1$ to $F$}
        \STATE $GDB (f, g_f) \lt (GDB(f, g_f).IL_{f, g_f} \cup id, (\bm{E}_{f, g_f}, \tau_{f, g_f})^*)$
        \label{code:mcdb-insert-wss-group}
    \ENDFOR
\ENDFOR  \label{code:mcdb-insert-wss-end}
\end{algorithmic}
\end{algorithm}
\framework allows users to update the database after bootstrapping.
However, after updating, the occurrences of involved elements will change.
To effectively protect the search pattern, we should ensure the elements in the same group always have the same occurrence.
\framework achieves that by updating dummy records.

\paragraph{Insert Query}
In \framework, the insert query is also performed with the cooperation of the user, IWS, and SSS.
The idea is that a number of dummy records will be generated and inserted with the real one to ensure all the elements in the same group always have the same occurrence.
The details are shown in Algorithm \ref{alg:mcdb-insert}.

Assume the real record to be inserted is $rcd=(e_1, \ldots, e_F)$.
The user encrypts it with \emph{RcdEnc}, and gets $(Ercd, seed, \bm{n}, Grcd)$ (Line \ref{code:mcdb-insert-user-enc}, Algorithm \ref{alg:mcdb-insert}).
For each $g_f \in Grcd$, the user gets $(\bm{E}_{f, g_f}, \tau_{f, g_f})^*$ of group $(f, g_f)$ and decrypts it.
Note that if $(f, g_f) \notin GDB$, the IWS returns $(\bm{E}_{f, g_f}, \tau_{f, g_f})^*$ of the closest group(s), instead of adding a new group.
That is, $e_f$ will belong to its closet group in this case.
The problem of adding new groups is that when the new groups contain less than $\lambda$ elements, adversaries could easily infer the search and access patterns within these groups.

The next step is to generate dummy records (Lines \ref{code:mcdb-insert-gendum-be}-\ref{code:mcdb-insert-gendum-end}).
The approach of generating dummy records depends on whether $e_f \in \bm{E}_{f, g_f}$, \ie whether $rcd$ introduces new element(s) that not belongs to $U$ or not.
If $e_f \in \bm{E}_{f, g_f}$, after inserting $rcd$, $O(e_f)$ will increase to $\tau_{f, g_f}+1$ automatically.
In this case, the occurrence of other elements in $\bm{E}_{f, g_f}$ should also be increased to $\tau_{f, g_f}+1$.
Otherwise, $O(e_f)$ will be unique in the database, and adversaries can tell if users are querying $e_f$ based on the size pattern.
To achieve that, $\gamma_f=|\bm{E}_{f, g_f}| -1$ dummy records are required for field $f$, and each of them contains an element in $\bm{E}_{f, g_f} \setminus e_f$.
If $e_f \notin \bm{E}_{f, g_f}$, $O(e_f)=0$ in $EDB$.
After inserting, we should ensure $O(e_f)=\tau_{f, g_f}$ since it belongs to the group $(f, g_f)$.
Thus, this case needs $\gamma_f=\tau_{f, g_f} -1$ dummy records for field $f$, and all of them are assigned with $e_f$ in field $f$.

Assume $W$ dummy records are required for inserting $rcd$, where $W=\max\{\gamma_f\}_{1 \leq f \leq F}$.
The user generates $W$ dummy records as mentioned above (`NULL' is used if necessary), and encrypts them with $RcdEnc$ (Lines \ref{code:mcdb-insert-encdum-be}-\ref{code:mcdb-insert-gendum-end}).
Meanwhile, the user adds each new element into the respective $\bm{E}_{f, g_f}$ if there are any, updates $\tau_{f, g_f}$, and re-encrypts $(\bm{E}_{f, g_f}, \tau_{f, g_f})$.
All the encrypted records are sent to the SSS and added into $EDB$ (Lines \ref{code:mcdb-insert-csp-be}-\ref{code:mcdb-insert-csp-end}).
All the $(seed, \bm{n})$ pairs and $(\bm{E}_{f, g_f}, \tau_{f, g_f})^*$ are sent to the IWS and inserted into $NDB$ and $GDB$ accordingly (Lines \ref{code:mcdb-insert-wss-be}-\ref{code:mcdb-insert-wss-end}).
Finally, to protect the access pattern, the shuffling and re-randomising operations over the involved groups will be performed between the IWS and RSS.

\paragraph{Delete Query}
Processing delete queries is straightforward.
Instead of removing records from the database, the user sets them to dummy by replacing their $tag$s  with random strings.
In this way, the occurrences of involved elements are not changed.
Moreover, the correctness of the search result is guaranteed.
However, only updating the tags of matched records leaks the access pattern of delete queries to the RSS.
Particularly, the RSS could keep a view of the latest database and check which records' tags were modified when the searched records are sent back for shuffling.
To avoid such leakage, in \framework, the user modifies the tags of all the searched records for each delete query.
Specifically, the SSS returns the identifiers of matched records and the tags of all searched records to the user.
For matched records, the user changes their tags to random strings directly.
Whereas, for each unmatched records, the user first checks if it is real or dummy, and then generates a proper new tag as done in Algorithm~\ref{alg:mcdb-enc}.
Likewise, the \emph{PreShuffle} and \emph{Shuffle} algorithms are performed between the IWS and RSS after updating all the tags.

However, if the system never removes records, the database will increase rapidly.
To avoid this, the admin periodically removes the records consisting of $F$ `NULL' from the database.
Specifically, the admin periodically checks if each element in each group is contained in one dummy record.
If yes, for each element, the admin updates one dummy record containing it to `NULL'.
As a consequence, the occurrence of all the elements in the same group will decrease, but still is the same.
When the dummy record only consists of `NULL', the admin removes it from the database.

\paragraph{Update query}
In \framework, update queries can be performed by deleting the records with old elements and inserting new records with new elements.

%% file: sections/proof.tex
\section{Security Analysis}
\label{sec:security}
In this section, we first analyse the leakage in \framework.
Second, we prove the patterns and forward and backward privacy are protected against the CSPs.

\subsection{Leakage of \framework}
Roughly speaking, given an initial database $DB$ and a sequence of queries $\bm{Q}$, the information leaked to each CSP in \framework can be defined as:
\begin{align*}
\mathcal{L} = \{\mathcal{L}_{\rm Setup}(DB), \{\mathcal{L}_{\rm Query}(Q_i)~or~\mathcal{L}_{\rm Update}(Q_i)\}_{Q_i \in \bm{Q}}\}
\end{align*}
where $\mathcal{L}_{\rm Setup}$, $\mathcal{L}_{\rm Query}$, and $\mathcal{L}_{\rm Update}$ represent the profiles leaked when setting up the system, executing queries and updating the database, respectively.
Specifically, $\mathcal{L}_{\rm Update}$ could be $\mathcal{L}_{\rm Insert}$ or $\mathcal{L}_{\rm Delete}$.
In the following, we analyse the specific information each CSP can learn from the received messages in each phase.

\paragraph{$\mathcal{L}_{\rm Setup}$}
When setting up the system, for the initial database $DB$, as shown in Algorithm \ref{alg:mcdb-setup}, the SSS gets the encrypted database $EDB$, and the IWS gets the group information $GDB$ and nonce information $NDB$.
In this phase, no data is sent to the RSS.
From $EDB$, the SSS learns the number of encrypted records $|EDB|$, the number of fields $F$, the length of each element $|e|$, and the length of tag $|tag|$.
From $NDB$ and $GDB$, the IWS learns $|NDB|$ ($|NDB|=|EDB|$), the length of each seed $|seed|$, the length of each nonce $|\bm{n}|$ ($|\bm{n}|=F|e|+|tag|$), the number of groups $|GDB|$, and the record identifiers $IL$ and $|(\bm{E}, \tau)^*|$ of each group.
In other words, the IWS learns the group information of each record in $EDB$.
Therefore, in this phase, the leakage $\mathcal{L}_{\rm Setup}(DB)$ learned by the SSS, IWS, and RSS can be respectively defined as:
\begin{align*}
    \mathcal{L}^{SSS}_{\rm Setup}(DB) = &\{|EDB|, \mathcal{L}_{rcd}\} \\
    \mathcal{L}^{IWS}_{\rm Setup}(DB) = &\{|NDB|, |\bm{n}|, |seed|, |GDB|,\\ & \{IL_{f, g}, |(\bm{E}_{f, g}, \tau)^*|\}_{(f, g) \in GDB}\} \\
    \mathcal{L}^{RSS}_{\rm Setup}(DB) = &\emptyset
\end{align*}
where $\mathcal{L}_{rcd}=\{F, |e|, |tag|\}$.

\paragraph{$\mathcal{L}_{\rm Query}$}
When processing queries, as mentioned in Algorithm \ref{alg:macdb-search} and \ref{alg:mcdb-shuffle}, the SSS gets the encrypted query $EQ$, $IL$ and encrypted nonces $EN$.
Based on them, the SSS can search over $EDB$ and gets the search result $SR$.
After shuffling, the SSS also gets the shuffled records $Ercds$ from the RSS.
From $\{EQ, IL, EN, SR, Ercds\}$, the SSS learns $\{Q.type, Q.f, |Q.e|, IL, |w|, |t|\}$, where $|t|=|seed|$.
In addition, the SSS can also infer the threshold $\tau$ ($\tau=|SR|$) and the number of distinct elements $|\bm{E}|$ ($|\bm{E}|=\frac{|IL|}{\tau}$) of the searched group.
The IWS only gets $(EQ.f, g, \eta)$ from the user, from which the IWS learns the searched field and group information of each query, and $|\eta|$ ($|\eta|=|Q.e|$).
The RSS gets the searched records $Ercds$, shuffled record identifies $IL'$, and new nonces $NN$ for shuffling and re-randomising.
From them, the RSS only learns $|Ercd|$ ($|Ercd|=|\bm{n}|$), $IL$ and $IL'$.
In summary,
\begin{align*}
	\mathcal{L}^{SSS}_{\rm Query}(Q) = &\{Q.f, Q.type, |Q.e|, Q.\bm{G}, |t|, |w|\} \\
    \mathcal{L}^{IWS}_{\rm Query}(Q) = &\{Q.f, g, |\eta|, |t|, |w|\} \\
    \mathcal{L}^{RSS}_{\rm Query}(Q) = &\{|Ercd|, IL, IL'\}
\end{align*}
where the group information $Q.\textbf{G}=\{g, IL, \tau, |\bm{E}|\}$.

\paragraph{$\mathcal{L}_{\rm Update}$}
Since different types of queries are processed in different manners, the SSS can learn if users are inserting, deleting or updating records, \ie $Q.type$.
As mentioned in Section \ref{subsec:update}, when inserting a real record, the user generates $W$ dummy ones, encrypts both the real and dummy records with $RcdEnc$, and sends them and their group information to the SSS and IWS, respectively.
Consequently, the SSS learns $W$, which represents the threshold or the number of elements of a group, and the IWS also learns the group information of each record.
Moreover, both the SSS and IWS can learn if the insert query introduces new elements that not belong to $U$ based on $|\bm{E}_{f, g}|$ or $|(\bm{E}_{f, g}, \tau_{f, g})^*|$.
The RSS only performs the shuffle operation.
Therefore, $\mathcal{L}_{\rm Insert}(Q)$ learned by each CSP is
\begin{align*}
	\mathcal{L}^{SSS}_{\rm Insert}(Q) = &\{W, \mathcal{L}_{rcd}\} \\
    \mathcal{L}^{IWS}_{\rm Insert}(Q) = &\{Grcd, \{|(\bm{E}_{f, g}, \tau_{f, g})^*|\}_{g_f \in Grcd}, W\} \\
    \mathcal{L}^{RSS}_{\rm Insert}(Q) = &\mathcal{L}^{RSS}_{\rm Query}(Q)
\end{align*}

Delete queries are performed as select queries in \framework, thus $\mathcal{L}_{\rm Delete} = \mathcal{L}_{\rm Query}$ for each CSP.

Above all,
\begin{align*}
	\mathcal{L}^{SSS} = &\{|EDB|, F, |e|, |tag|, |t|, |w|\\ &\{\{Q_i.f, Q_i.type, |Q_i.e|,  Q_i.\bm{G}\} ~or~ W\}_{Q_i \in \textbf{Q}}\} \\
    \mathcal{L}^{IWS} = &\{|NDB|, |GDB|, |\bm{n}|, |seed|, |w|, \\ &\{IL_{f, g}, |(\bm{E}_{f, g}, \tau)^*|\}_{(f, g) \in GDB}, \{\{Q_i.f, Q_i.g, |Q_i.e|\} \\ & ~or~ \{Grcd, \{|(\bm{E_{f, g}}, \tau_{f, g})^*|\}_{g_f \in Grcd}, W\}\}_{Q_i \in \textbf{Q}} \} \\
    \mathcal{L}^{RSS} = &\{|Ercd|, \{IL, IL'\}_{Q_i \in \bm{Q}}\}
\end{align*}


\subsection{Proof of Security}
Given the above leakage definition for each CSP, in this part, we prove adversaries do not learn anything beyond $\mathcal{L}^{csp}$ by compromising the CSP $csp$, where $csp$ could be the SSS, IWS, or RSS.
It is clear that adversaries cannot infer the search, access and size patterns, and forward and backward privacy of queries within a group from $\mathcal{L}^{csp}$.
Therefore, proving \framework only leaks $\mathcal{L}^{csp}$ to $csp$ indicates \framework protects the patterns and ensures forward and backward privacy within groups.

To prove \framework indeed only leaks $\mathcal{L}^{csp}$ to $csp$, we follow the typical method of using a real-world versus ideal-world paradigm \cite{Bost:2017:forward,CashJJJKRS14,KamaraPR12}.
The idea is that first we assume the CSP $csp$ is compromised by a Probabilistic Polynomial-Time (PPT) honest-but-curious adversary $\mathcal{A}$ who follows the protocol honestly as done by $csp$, but wants to learn more information by analysing the received messages and injecting malicious records.
Second, we build two experiments: $\textbf{Real}^{\rm \Pi}_{\mathcal{A}}(k)$ and $\textbf{Ideal}^{\rm \Pi}_{\mathcal{A, S}, \mathcal{L}}(k)$, where $\Pi$ represents \framework.
In $\textbf{Real}^{\rm \Pi}_{\mathcal{A}}(k)$, all the messages sent to $\mathcal{A}$ are generated as specified in \framework.
Whereas, in $\textbf{Ideal}^{\rm \Pi}_{\mathcal{A, S}, \mathcal{L}}(k)$, all the messages are generated by a PPT simulator $\mathcal{S}$ that only has access to $\mathcal{L}^{csp}$.
That is, $\mathcal{S}$ ensures $\mathcal{A}$ only learns the information defined in $\mathcal{L}^{csp}$ from received messages in $\textbf{Ideal}^{\rm \Pi}_{\mathcal{A, S}, \mathcal{L}}(k)$.
In the game, $\mathcal{A}$ chooses an initial database, triggers $Setup$, and adaptively issues \emph{select}, \emph{insert}, and \emph{delete} queries of its choice.
In response, either $\textbf{Real}^{\rm \Pi}_{\mathcal{A}}(k)$ or $\textbf{Ideal}^{\rm \Pi}_{\mathcal{A, S}, \mathcal{L}}(k)$ is invoked to process the database and queries.
Based on the received messages, $\mathcal{A}$ distinguishes if they are generated by  $\textbf{Real}^{\rm \Pi}_{\mathcal{A}}(k)$ or $\textbf{Ideal}^{\rm \Pi}_{\mathcal{A, S}, \mathcal{L}}(k)$.
If $\mathcal{A}$ cannot distinguish that with non-negligible advantage, it indicates $\textbf{Real}^{\rm \Pi}_{\mathcal{A}}(k)$ has the same leakage profile as $\textbf{Ideal}^{\rm \Pi}_{\mathcal{A, S}, \mathcal{L}}(k)$.

%
\begin{mydef}
We say the dynamic SSE scheme is $\mathcal{L}$-adaptively-secure against the CSP $csp$, with respect to the leakage function $\mathcal{L}^{csp}$, if for any PPT adversary issuing a polynomial number of queries, there exists a PPT simulator $\mathcal{S}$, such that $\textbf{Real}^{\rm \Pi}_{\mathcal{A}}(k)$ and $\textbf{Ideal}^{\rm \Pi}_{\mathcal{A, S}, \mathcal{L}}(k)$ are distinguishable with negligible probability $\textbf{negl}({k})$.
\end{mydef}

Herein, we acknowledge that \framework leaks the group information of queries and records and leaks whether the elements involved in select, insert and delete queries belong to $U$ or not.  For clarity, in the proof we assume there is only one group in each field, and omit the group processing. Moreover, we assume all the queries issued by $\mathcal{A}$ only involve elements in $U$.
In this case, the leakage learned by each CSP can be simplified into:
\begin{align*}
	\mathcal{L}^{SSS} = &\{|EDB|, F, |e|, |tag|, |t|, |w|\\ &\{\{Q_i.f, Q_i.type, |Q_i.e|\} ~or~ W\}_{Q_i \in \textbf{Q}}\} \\
    \mathcal{L}^{IWS} = &\{|NDB|, |\bm{n}|, |seed|, |w|, \{|(\bm{E}_{f}, \tau)^*|\}_{f \in [1, F]}, \\ &\{Q_i.f~or~W\}_{Q_i \in \textbf{Q}} \} \\
    \mathcal{L}^{RSS} = &\{|Ercd|, \{IL'\}_{Q_i \in \bm{Q}}\}
\end{align*}


\begin{theorem}\label{the::SSS}
If $\Gamma$ is secure PRF, $\pi$ is a secure PRP, and $H'$ is a random oracle, \framework is a $\mathcal{L}$-adaptively-secure dynamic SSE scheme against the SSS.
\end{theorem}

\begin{proof}
To argue the security, as done in \cite{Bost:2017:forward,CashJJJKRS14,KamaraPR12}, we prove through a sequence of games.
The proof begins with $\textbf{Real}^{\rm \Pi}_{\mathcal{A}}(k)$, which is exactly the real protocol, and constructs a sequence of games that differ slightly from the previous game and show they are indistinguishable.
Eventually we reach the last game $\textbf{Ideal}^{\rm \Pi}_{\mathcal{A}, \mathcal{S}}(k)$, which is simulated by a simulator $\mathcal{S}$ based on the defined leakage profile $\mathcal{L}^{SSS}$.
By the transitive property of the indistinguishability, we conclude that $\textbf{Real}^{\rm \Pi}_{\mathcal{A}}(k)$ is indistinguishable from $\textbf{Ideal}^{\rm \Pi}_{\mathcal{A}, \mathcal{S}}(k)$ and complete our proof.
Since $RcdDec$ is unrelated to CSPs, it is omitted in the games.
\begin{algorithm}[!htp]
\scriptsize
\caption{$\textbf{Real}^{\rm \Pi}_{\mathcal{A}}(k).RcdEnc(rcd, flag)$ $\|$ \fbox{$\mathcal{G}_1$}, \fbox{$\mathcal{G}_2$}, \fbox{$\mathcal{G}_3$}}
\label{proof::h1::enc}
\begin{algorithmic}[1]
\STATE $seed \stackrel{\$}{\lt} \{0,1\}^{|seed|}$
    \STATE $\bm{n} \lt \Gamma_{s_2} (seed)$ $\vartriangleleft$ \fbox{$\mathcal{G}_1$: $\bm{n} \lt \bm{Nonce}[seed]$}, where $\bm{n}= \ldots \parallel n_f \parallel \ldots \parallel n_{F+1}$, $|n_f|=|e|$ and $|n_{F+1}|=|tag|$
\FOR {each element $e_f \in rcd$}
    \STATE ${e^*_f} \lt Enc_{s_1}(e_f) \oplus n_f$  $\vartriangleleft$  {\fbox{$\mathcal{G}_3$: $e^*_f \lt \{0, 1\}^{|e|}$}}
\ENDFOR
\IF {$flag=1$}
    \STATE $S \stackrel{\$}{\leftarrow} \{0,1\}^{|e|}$
    \STATE $tag \leftarrow (H_{s_1}(S)||S )\oplus n_{F+1}$
    $\vartriangleleft$  {\fbox{$\mathcal{G}_3$: $tag \lt \{0, 1\}^{|tag|}$}}
\ELSE
    \STATE $tag \stackrel{\$}{\leftarrow} \{0,1\}^{|tag|}$
\ENDIF
\RETURN $Ercd=(e^*_1, \ldots, e^*_F, tag)$ and $(seed, \bm{n})$ 
\end{algorithmic}
\end{algorithm}
\begin{algorithm}[!htp]
\scriptsize
\caption{$\textbf{Real}^{\rm \Pi}_{\mathcal{A}}(k).Query(Q)$ $\|$ \fbox{$\mathcal{G}_1$}, \fbox{$\mathcal{G}_2$}, \fbox{$\mathcal{G}_3$}}
\label{h1:macdb-search}
\begin{algorithmic}[1]
	\STATE \underline{User: $QueryEnc(Q)$}
    \STATE $EQ.type \lt Q.type$, $EQ.f \lt Q.f$
    \STATE $\eta \stackrel{\$}{\lt} \{0,1\}^{|e|}$
    \STATE $EQ.e^* \leftarrow Enc_{s_1}(Q.e) \oplus \eta$  $\vartriangleleft$  {\fbox{$\mathcal{G}_3$: $EQ.e^* \lt \{0, 1\}^{|e|}$}}
    \STATE Send $EQ=(type, f, op, e^*)$ to the SSS
    \STATE Send $(EQ.f, \eta)$ to the IWS

~\\

\STATE \underline{IWS: $NonceBlind(EQ.f, \eta)$}
    \STATE $EN \leftarrow \emptyset$
    \STATE \fbox{$\mathcal{G}_2$: Randomly put $\tau_f$ record identifiers into $\bm{I}$}
    \FOR {each $id \in NDB$ }
        \STATE $(seed, \bm{n}) \leftarrow NDB(id)$, where $\bm{n}= \ldots ||n_{EQ.f}|| \ldots $ and $|n_{EQ.f}|=|\eta|$    $\vartriangleleft$   \fbox{Deleted in $\mathcal{G}_2$}
        \STATE  ${w} \lt H'(n_{EQ.f} \oplus \eta)$   $\vartriangleleft$    \fbox{$\mathcal{G}_2$:~~\begin{minipage}[c][1.0cm][t]{4.5cm}{\textbf{if} $id \in \bm{I}$ \\ $w_{id} \lt H'(EDB(id, EQ.f) \oplus EQ.e^*)$ \\ \textbf{else} \\ $w_{id} \lt \{0, 1\}^{|w|}$ } \end{minipage}}
        \STATE  $t \lt \eta \oplus seed$   $\vartriangleleft$  {\fbox{$\mathcal{G}_3$: $t \lt \{0, 1\}^{|seed|}$}}
        \STATE  $EN(id) \leftarrow (w, t)$
    \ENDFOR
    \STATE Send the encrypted nonce set $EN=((w, t), \ldots)$ to the SSS

~\\
\STATE \underline{SSS: $Search(EQ, EN)$}
    \STATE $SR \leftarrow \emptyset$
    \FOR {each $id \in EDB$}
        \IF {$H'(EDB(id, EQ.f) \oplus EQ.e^*) = EN(id).w$}
             \STATE $SR \leftarrow SR \cup (EDB(id), EN(id).t)$
        \ENDIF
    \ENDFOR
    \STATE Send the search result $SR$ to the user
		
	\end{algorithmic}
\end{algorithm}
\begin{algorithm}[!ht]
\scriptsize
\caption{$\textbf{Real}^{\rm \Pi}_{\mathcal{A}}(k).Shuffle()$ $\|$ \fbox{$\mathcal{G}_1$}, \fbox{$\mathcal{G}_2$}, \fbox{$\mathcal{G}_3$}}
\label{h1:mcdb-shuffle}
\begin{algorithmic}[1]
\STATE \underline{IWS: $PreShuffle()$}
	\STATE {$IL' \leftarrow \pi(NDB)$}
	\FOR {each $id \in IL'$}
		\STATE $seed \stackrel{\$}{\leftarrow} \{0, 1\}^{|seed|}$
        \STATE $\bm{n}' \lt \Gamma_{s_2}(seed)$  $\vartriangleleft$ \fbox{$\mathcal{G}_1$: $\bm{n}' \lt \bm{Nonce}[seed]$} 
		\STATE $NN(id) \leftarrow NDB(id).\bm{n} \oplus \bm{n}' $   $\vartriangleleft$ \fbox{$\mathcal{G}_3$: $NN(id) \lt \{0, 1\}^{|\bm{n}|}$}
		\STATE $NDB(id)\leftarrow (seed, \bm{n}')$
	\ENDFOR
	\STATE Send $(IL', NN)$ to the RSS.

~\\~
\STATE{\underline{RSS: $Shuffle(Ercds, IL', NN)$}}
	\STATE Shuffle $Ercds$ based on $IL'$
	\FOR {each $id \in IL'$}
	\STATE {$Ercds(id) \leftarrow Ercds(id) \oplus NN(id)$}   $\vartriangleleft$ \fbox{$\mathcal{G}_3$: $Ercds(id) \lt \{0, 1\}^{|\bm{n}|}$}
	\ENDFOR
	\STATE Send $Ercds$ to the SSS.
\end{algorithmic}
\end{algorithm}
%
\noindent\textbf{Game $\mathcal{G}_1$}: Comparing with $\textbf{Real}^{\rm \Pi}_{\mathcal{A}}(k)$, the difference in $\mathcal{G}_1$ is that the PRF $\Gamma$ for generating nonces, in $RcdEnc$ and $PreShuffle$ algorithms, is replaced with a mapping \textbf{Nonce}.
Specifically, as shown in Algorithm \ref{proof::h1::enc} and \ref{h1:mcdb-shuffle}, for each unused $seed$ (the length of seed is big enough), a random string of length $F|e|+|tag|$ is generated as the nonce, stored in \textbf{Nonce}, and then reused thereafter.
This means that all of the $\bm{n}$ are uniform and independent strings.
In this case, the adversarial distinguishing advantage between $\textbf{Real}^{\rm \Pi}_{\mathcal{A}}(k)$ and $\mathcal{G}_1$ is exactly the distinguishing advantage between a truly random function and PRF.
Thus, this change made negligible difference between between $\textbf{Real}^{\rm \Pi}_{\mathcal{A}}(k)$ and $\mathcal{G}_1$, \ie
\[
\centerline { $|{\rm Pr}[\textbf{Real}^{\rm \Pi}_{\mathcal{A}}(k)=1] - {\rm Pr}[\mathcal{G}_1=1]| \leq \textbf{negl}({k})$}
\]
where ${\rm Pr}[\mathcal{G}=1]$ represents the probability of that the messages received by $\mathcal{A}$ are generated by $\mathcal{G}$.

\noindent\textbf{Game $\mathcal{G}_2$}:
From $\mathcal{G}_1$ to $\mathcal{G}_2$, $w$ is replaced with a random string, rather than generated via $H'$.
However, it is necessary to ensure $\mathcal{A}$ gets $\tau_f$ matched records after searching over $EDB$, since that is the leakage $\mathcal{A}$ learns, where $\tau_f$ is the threshold of the searched field.
To achieve that, the experiment randomly picks $\tau_f$ witnesses and programs their values.
Specifically, as shown in Algorithm \ref{h1:macdb-search}, the experiment first randomly picks a set of record identifiers $\bm{I}$, where $|\bm{I}|=\tau_f$.
Second, for each identifier $id \in \bm{I}$, the experiment programs $w_{id} \lt H'(EDB(id, EQ.f) \oplus EQ.e^*)$.
By doing so, the records identified by $\bm{I}$ will match the query.
For the identifier $id \notin \bm{I}$, $w_{id} \lt \{0, 1\}^{|w|}$.

The only difference between $\mathcal{G}_2$ and $\mathcal{G}_1$ is the generation of $w$.
In the following, we see if $\mathcal{A}$ can distinguish the two games based on $w$.
In $\mathcal{G}_2$,
\begin{align}\notag
  For~id \in \bm{I}, ~w_{id} &\lt H'(EDB(id, EQ.f) \oplus EQ.e^*)  \\ \notag
  For~id \notin \bm{I}, ~w_{id} &\lt \{0, 1\}^{|w|}
\end{align}
Recall that in $\mathcal{G}_1$.
\begin{equation}\notag
  w_{id} \lt H'(n_{EQ.f} \oplus \eta)
\end{equation}
In $\mathcal{G}_1$, $n_{EQ.f}$ and $\eta$ are random strings.
In $\mathcal{G}_2$, due to the one-time pad encryption in $RcdEnc$ and $QueryEnc$, $EDB(id, EQ.f)$ and $EQ.e^*$ are indistinguishable from random strings.
Thus, we can say for $id \in \bm{I}$ $w_{id}$ is generated in the same way as done in $\mathcal{G}_1$. 
For $id \notin \bm{I}$, $w_{id}$ is a random string in $\mathcal{G}_2$, whereas in $\mathcal{G}_1$  $w_{id}$ is generated by deterministic $H'$. 
It seems $\mathcal{A}$ could easily distinguish $\mathcal{G}_2$ and $\mathcal{G}_1$, since $\mathcal{G}_1$ outputs the same $w$ for the same input, whereas $\mathcal{G}_2$ does not.
Indeed, in $\mathcal{G}_1$ the inputs to $H'$, $n_{EQ.f}$ and $\eta$, are random strings, thus the probability of getting the same input for $H'$ is negligible, making $H'$ indistinguishable from a uniform sampling. 
Thus, in both cases $w_{id}$ in $\mathcal{G}_2$ is indistinguishable from $w_{id}$ in $\mathcal{G}_1$. 

Next, we discuss if $\mathcal{A}$ can distinguish the two games based on $SR$.
The leakage of $SR$ includes the identifier of each matched records and $|SR|$.
Due to the padding, $|SR|=|\bm{I}|=\tau_f$, which means the two games are indistinguishable based on $|SR|$.
In $\mathcal{G}_1$, the identifiers of matched records are determined by the shuffle operations performed for the previous query.
In $\mathcal{G}_2$, the identifiers of matched records are randomly picked.
Thus, the distinguishing advantage between $\mathcal{G}_1$ and $\mathcal{G}_2$ based on the identifiers is exactly the distinguishing advantage between a truly random permutation and PRP, which is negligible.

Above all, we have
\[
\centerline { $|{\rm Pr}[\mathcal{G}_2=1] - {\rm Pr}[\mathcal{G}_1=1]| \leq \textbf{negl}(k)$}
\]

\noindent\textbf{Game $\mathcal{G}_3$}:
The difference between $\mathcal{G}_2$ and $\mathcal{G}_3$ is that all the XORing operations, such as the generation of $e^*$, $Q.e^*$, and $t$, are replaced with randomly sampled strings (The details are shown in Algorithms \ref{proof::h1::enc}, \ref{h1:macdb-search}, and \ref{h1:mcdb-shuffle}).
Since sampling a fixed-length random string is indistinguishable from the one-time pad encryption,
we have
\[
\centerline { ${\rm Pr}[\mathcal{G}_3=1] = {\rm Pr}[\mathcal{G}_2=1] $}
\]

%
\begin{algorithm}[!htp]
\scriptsize
\caption{$\mathcal{S}.RcdEnc(\mathcal{L}_{rcd}$)}
\label{proof::ideal::enc}
\begin{algorithmic}[1]
	\STATE \tgrey{$seed \stackrel{\$}{\lt} \{0,1\}^{|seed|}$}
	\STATE \tgrey{$\bm{n} \lt \bm{Nonce}[seed]$}
    \FOR {each $f \in [1, F]$}
        \STATE ${e^*_f} \lt \{0, 1\}^{|e|}$
    \ENDFOR
	\STATE 	{$tag \stackrel{\$}{\leftarrow} \{0,1\}^{|tag|}$}
	\RETURN {$Ercd=(e^*_1, \ldots, e^*_F, tag)$} and \tgrey{$(seed, \bm{n})$} 
\end{algorithmic}
\end{algorithm}

%
\begin{algorithm}[!htp]
\scriptsize
\caption{$\mathcal{S}.Query(\mathcal{L}^{SSS}_{Query})$}
\label{ideal:macdb-search}
\begin{algorithmic}[1]
	\STATE \underline{User: $QueryEnc(\mathcal{L}^{SSS}_{Query})$}
    \STATE $EQ.type \lt Q.type$, $EQ.f \lt Q.f$
    \STATE \tgrey{$\eta \stackrel{\$}{\lt} \{0,1\}^{| e |}$}
    \STATE $EQ.e^* \leftarrow \{0, 1\}^{|e|}$
    \STATE Send $EQ=(type, f, op, e^*)$ to the SSS
    \STATE \tgrey{Send $(EQ.f, \eta)$ to the IWS}

~\\

\STATE \underline{IWS: $NonceBlind(\mathcal{L}^{SSS}_{Query})$}
    \STATE $EN \leftarrow \emptyset$
    \STATE Randomly put $\tau_f$ records identifers into $\bm{I}$
    \FOR {each $id \in [1, |EDB|]$ }
        \IF{$id \in \bm{I}$}
            \STATE  ${w} \lt H'(EDB(id, EQ.f) \oplus EQ.e^*)$
        \ELSE
            \STATE  ${w} \lt \{0, 1\}^{|w|}$
        \ENDIF
        \STATE  $t \lt \{0, 1\}^{|seed|}$
        \STATE  $EN(id) \leftarrow (w, t)$
    \ENDFOR
    \STATE Send the encrypted nonce set $EN=((w, t), \ldots)$ to the SSS

~\\
\STATE \underline{SSS: $Search(EQ, EN)$}
    \STATE $SR \leftarrow \emptyset$
    \FOR {each $id \in EDB$}
        \IF {$H'(EDB(id, EQ.f) \oplus EQ.e^*) = EN(id).w$ }
             \STATE $SR \leftarrow SR \cup (EDB(id), EN(id).t)$
        \ENDIF
    \ENDFOR
    \STATE Send the search result $SR$ to the user
		
	\end{algorithmic}
\end{algorithm}

\begin{algorithm}
\scriptsize
\caption{$\mathcal{S}.Shuffle(\mathcal{L}^{SSS}_{Query})$}
\label{proof::ideal::shuffle}
\begin{algorithmic}[1]
\STATE \underline{IWS: $PreShuffle()$}
    \STATE \tgrey{$IL' \leftarrow RP (NDB)$}
    \FOR {\tgrey{each $id \in IL'$}}
        \STATE \tgrey{$seed \stackrel{\$}{\leftarrow} \{0, 1\}^{|seed|}$}
        \STATE \tgrey{$\bm{n'} \leftarrow  \bm{Nonce}[seed]$}
        \STATE \tgrey{$NN(id) \leftarrow \{0, 1\}^{|\bm{n}|}$}
        \STATE \tgrey{$NDB(id)\leftarrow (seed, \bm{n'})$}
    \ENDFOR
    \STATE \tgrey{Send $(IL', NN)$ to the RSS.}

~\\

\STATE{\underline{RSS: $Shuffle(\mathcal{L}^{SSS}_{Query})$}}
    \STATE \tgrey{Shuffle $Ercds$ based on $IL'$}
    \FOR {each $id \in IL$}
        \STATE $Ercds(id) \leftarrow \{0, 1\}^{|\bm{n}|} $
    \ENDFOR
    \STATE Send $Ercds$ to the SSS.
\end{algorithmic}
\end{algorithm}
%

\noindent\textbf{$\textbf{Ideal}^{\rm \Pi}_{\mathcal{A}, \mathcal{S}, \mathcal{L}}(k)$}:
From $\mathcal{G}_3$ to the final game, we just replace the inputs to $RcdEnc$, $Query$ and $Shufle$ algorithms with $\mathcal{L}^{SSS}$.
Moreover, for clarity, we fade the operations unrelated to the SSS.
From Algorithms \ref{proof::ideal::enc}, \ref{ideal:macdb-search}, and \ref{proof::ideal::shuffle}, it is easy to observe that the messages sent to $\mathcal{A}$, \ie $\{Ercd, EQ, EN, Ercds\}$, can be easily simulated by only relying on $\mathcal{L}^{SSS}$.
Here we have:

\[
{\rm Pr}[\textbf{Ideal}^{\rm \Pi}_{\mathcal{A},\mathcal{S}, \mathcal{L}}(k)=1]={\rm Pr}[\mathcal{G}_3=1]
\]

By combining all the distinguishable advantages above, we get:
\[
|{\rm Pr}[\textbf{Ideal}^{\rm \Pi}_{\mathcal{A},\mathcal{S}, \mathcal{L}}(k)=1]- {\rm Pr}[\textbf{Real}^{\rm \Pi}_{\mathcal{A}}(k)=1]| < \textbf{negl}(k)
\]

In $\textbf{Ideal}^{\rm \Pi}_{\mathcal{A},\mathcal{S}, \mathcal{L}}(k)$, $\mathcal{A}$ only learns $\mathcal{L}_{rcd}$ and $\mathcal{L}^{SSS}_{Query}$.
The negligible advantage of a distinguisher between $\textbf{Real}^{\rm \Pi}_{\mathcal{A}}(k)$ and $\textbf{Ideal}^{\rm \Pi}_{\mathcal{A},\mathcal{S}, \mathcal{L}}(k)$ indicates that \framework also only leaks $\mathcal{L}_{rcd}$ and $\mathcal{L}^{SSS}_{Query}$.

Although the above simulation does not include the $Setup$ and updating phases, it is clear that the two phases mainly rely on $RcdEnc$ algorithm, which has been proved only leaks $\mathcal{L}_{rcd}$ to the SSS.
From $Setup$ phase, $\mathcal{A}$ only gets $EDB$, and each record in $EDB$ is encrypted with $RcdEnc$.
Thus, $\mathcal{A}$ only learns $|EDB|$ and $\mathcal{L}_{rcd}$ in this phase.
Similarly, $\mathcal{A}$ only gets $W+1$ encrypted records in $Insert$ algorithm.
Therefore, in addition to $\mathcal{L}_{rcd}$, it only learns $W$, which is equal to $|\bm{E}|$ or $\tau$ of a group.
For delete queries, $\mathcal{A}$ learns $\mathcal{L}^{SSS}_{Query}$ since they are first performed as select queries.
As proved above, the tags are indistinguishable from random strings, meaning the returned tags do not leak additional information.
\end{proof}

\begin{theorem}\label{the::IWS}
If $ENC$ is semantically secure, \framework is a $\mathcal{L}$-adaptively-secure dynamic SSE scheme against the IWS.
\end{theorem}

\begin{proof}
Herein, we also assume all the records are in one group.
In this case, the IWS gets $(seed, \bm{n})$ for each record and $(\bm{E}_f, \tau_f)^*$ for each field when setting up the database, gets $(Q.f, \eta)$ when executing queries, and gets updated $(\bm{E}_f, \tau_f)^*$ when inserting records.
Note that the IWS can generate $\bm{n}$ by itself since it has the key $s_2$.
Considering $seed$ and $\eta$ are random strings, from $\textbf{Real}^{\rm \Pi}_{\mathcal{A}}(k)$ to $\textbf{Ideal}^{\rm \Pi}_{\mathcal{A}, \mathcal{S}, \mathcal{L}}(k)$ we just need one step.
Specifically, given $\mathcal{L}^{IWS}$, in $\textbf{Ideal}^{\rm \Pi}_{\mathcal{A}, \mathcal{S}, \mathcal{L}}(k)$, $\mathcal{S}$ just needs to simulate $(\bm{E}_f, \tau_f)^*$ with $|(\bm{E}_f, \tau_f)^*|$-bit random strings in $Setup$ and $Insert$ algorithms.
Given $ENC$ is semantically secure, we have
\[
|{\rm Pr}[\textbf{Ideal}^{\rm \Pi}_{\mathcal{A},\mathcal{S}, \mathcal{L}}(k)=1]- {\rm Pr}[\textbf{Real}^{\rm \Pi}_{\mathcal{A}}(k)=1]| < \textbf{negl}(k)
\]
\end{proof}

\begin{theorem}\label{the::RSS}
\framework is a $\mathcal{L}$-adaptively-secure dynamic SSE scheme against the RSS.
\end{theorem}

\begin{proof}
In \framework, the RSS is only responsible for shuffling and re-randomising records after each query.
For the shuffling and re-randomising, it gets encrypted records, $IL'$ and $NN$.
Here we also just need one step to reach $\textbf{Ideal}^{\rm \Pi}_{\mathcal{A}, \mathcal{S}, \mathcal{L}}(k)$.
Given $\mathcal{L}^{RSS}$, as done in the above \textbf{Game $\mathcal{G}_3$}, $\mathcal{S}$ needs to replace $e^*$ and $tag$ in $RcdEnc$ with $|e|$-bit and $|tag|$-bit random strings respectively and simulate each element in $NN$ with a $|Ercd|$-bit random string in $PreShuffle$.
As mentioned, since sampling a fixed-length random string is indistinguishable from the one-time pad encryption, we have
\[
|{\rm Pr}[\textbf{Ideal}^{\rm \Pi}_{\mathcal{A},\mathcal{S}, \mathcal{L}}(k)=1] = {\rm Pr}[\textbf{Real}^{\rm \Pi}_{\mathcal{A}}(k)=1]|
\]
\end{proof}

%% file: sections/performance.tex
\section{Performance Analysis}
\label{sec:MCDB-perf}
We implemented \framework in C using MIRACL 7.0 library for cryptographic primitives.
The performance of all the entities was evaluated on a desktop machine with Intel i5-4590 3.3 GHz 4-core processor and 8GB of RAM.
We evaluated the performance using TPC-H benchmark \cite{TPC:2017:h}, and tested equality queries with one singe predicates over `O\_CUSTKEY' field in `ORDERS' table.
In the following, all the results are averaged over $100$ trials.

\subsection{Group Generation}
\begin{table}[h]
\scriptsize
  \centering
    \caption{The storage overhead with different numbers of groups}
  \begin{tabular}{|l|l|l|l|l|}
     \hline
     \#Groups  &\#Dummy records &\#Elements in a group & \#Records in a group  \\ \hline
     1         & 2599836 & 99996 & =4099836  \\
     10        & 2389842 & 10000 & $\approx$38000 \\
     100       & 1978864 & 1000  & $\approx$35000  \\
     1000      & 1567983 & 100   & $\approx$3000  \\
     10000     & 1034007 & 10    & $\approx$240 \\
     \hline
   \end{tabular}
  \label{Tbl:oblidb-storage-perf}
\end{table}
In `ORDERS' table, all the `O\_CUSTKEY' elements are integers.
For simplicity, we divided the records into groups by computing $e~ mod~ b$ for each element $e$ in `O\_CUSTKEY' field.
Specifically, we divide the records into 1, 10, 100, 1000, and 10000 groups by setting $b=$1, 10, 100, 1000,  and 10000, respectively.

Table \ref{Tbl:oblidb-storage-perf} shows the number of required records and the number of elements included in a group when dividing the database into 1, 10, 100, 1000, and 10000 groups.
In particular, when all the records are in one group, $2599836$ dummy records are required in total, $\lambda=99996$, and the CSP has to search $4099836$ records for each query.
When we divide the records into more groups, fewer dummy records are required, fewer records will be searched by the CSP, but fewer elements will be contained in each group.
When there are $10000$ groups, only $1034007$ dummy records are required totally, $\lambda=10$, and the CSP just needs to search around $240$ records for each query.

\subsection{Query Latency}
\begin{figure}[htp]
\centering
\includegraphics[width=0.32\textwidth]{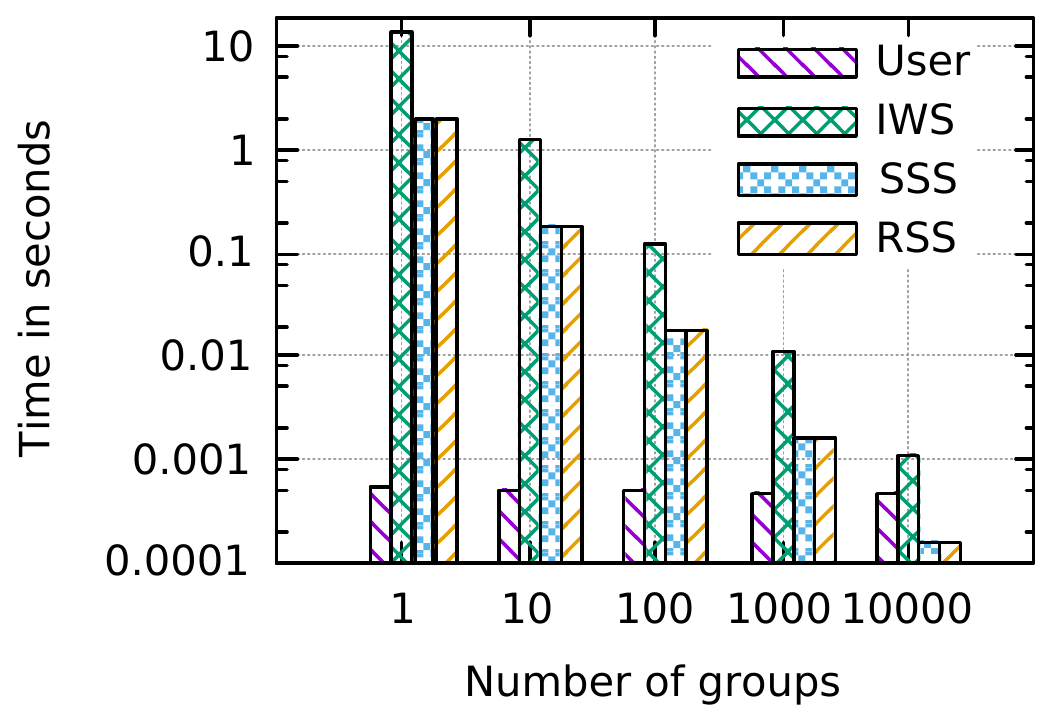}\\
\caption{The overhead on different entities with different group numbers}
\label{Fig:mcdb-entity}
\end{figure}
An important aspect of an outsourced service is that most of the intensive computations should be off-loaded to the CSPs.
To quantify the workload on each of the entities, we measured the latency on the user, IWS, SSS, and RSS for processing the query with different numbers of groups.
The results are shown in Fig.~\ref{Fig:mcdb-entity}.
We can notice that the computation times taken on the IWS, SSS, and RSS are much higher than that on the user side when there are less than 10000 groups.
In the following, we will discuss the performance on CSPs and the user in details.

\subsubsection{Overhead on CSPs}
\begin{figure}[htp]
\centering
\includegraphics[width=0.32\textwidth]{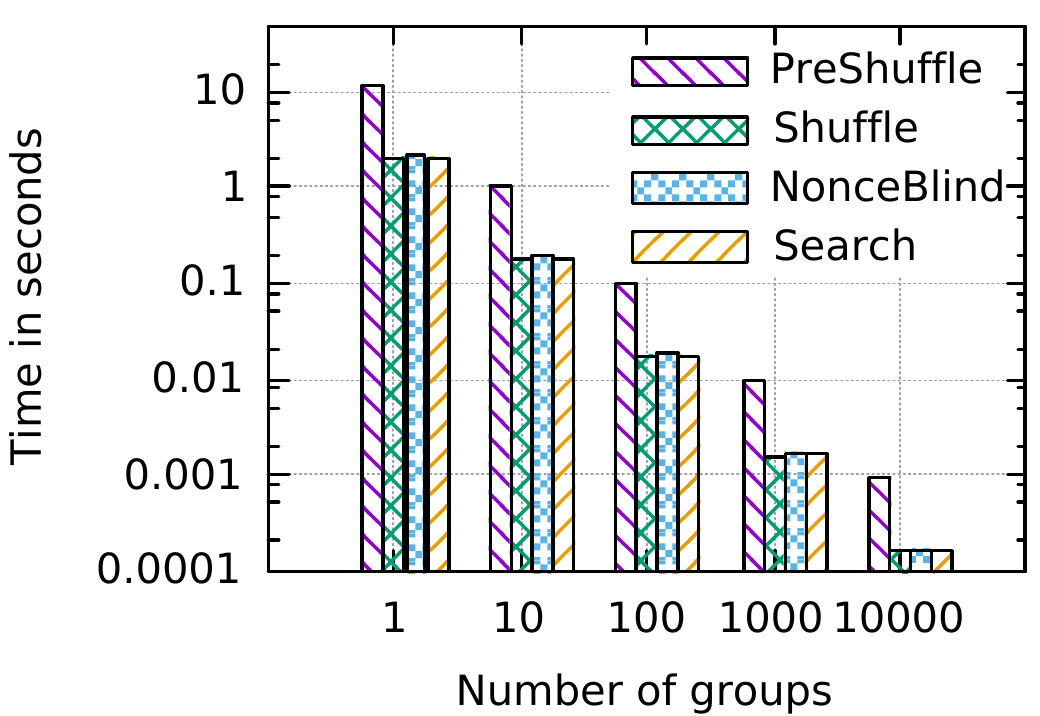}\\
\caption{The overhead on CSPs with different group numbers}
\label{Fig:mcdb-csps}
\end{figure}
Fig. \ref{Fig:mcdb-csps} shows the performance of the operations running in the CSPs when increasing the number of groups.
Specifically, in \framework, the IWS runs \emph{NonceBlind} and \emph{PreShuffle}, the SSS runs \emph{Search}, and the RSS runs \emph{Shuffle} algorithms.
We can notice that the running times of all the four operations reduce when increasing the number of groups.
The reason is that \framework only searches and shuffles a group of records for each query.
The more groups, the fewer records in each group for searching and shuffling.
Thanks to the efficient XOR operation, even when $g=1$, \ie searching the whole database (contains 4099836 records in total), \emph{NonceBlind}, \emph{Search}, and \emph{Shuffle} can be finished in around $2$ seconds.
\emph{PreShuffle} is the most expensive operation in \framework, which takes about $11$ seconds when $g=1$.
Fortunately, in \emph{PreShuffle}, the generation of new nonces (\ie Lines \ref{code:mcdb-seed'}-\ref{code:mcdb-gn'} in Algorithm \ref{alg:mcdb-shuffle}) is not affected by the search operation, thus they can be pre-generated.
By doing so, \emph{PreShuffle} can be finished in around $2.4$ seconds when $g=1$.

\subsubsection{Overhead on Users}
\begin{figure}[htp]
\centering
\includegraphics[width=0.32\textwidth]{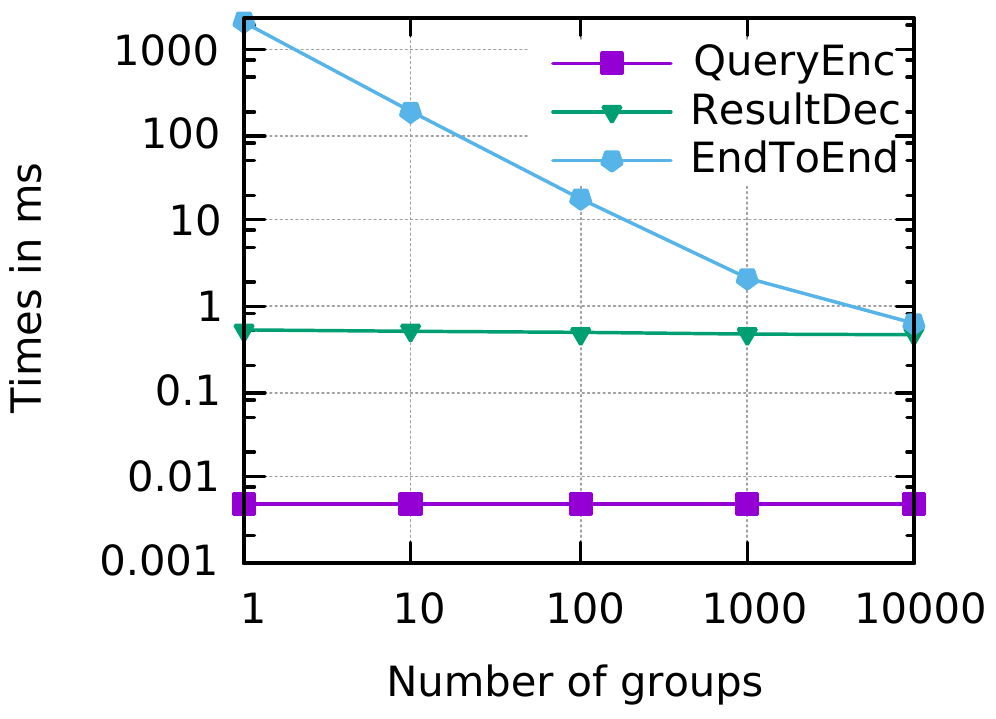}\\
\caption{The overhead on the user with different group numbers}
\label{Fig:mcdb-user}
\end{figure}
In \framework, the user only encrypts queries and decrypts results.
In Fig. \ref{Fig:mcdb-user}, we show the effect on the two operations when we change the number of groups.
The time for encrypting the query does not change with the number of groups.
However, the time taken by the result decryption decreases slowly when increasing the number of groups.
For recovering the required records, in \framework, the user first filters out the dummy records and then decrypts the real records.
Therefore, the result decryption time is affected by the number of returned real records as well as the dummy ones.
In this experiment, the tested query always matches 32 real records.
However, when changing the number of groups, the number of returned dummy records will be changed.
Recall that, the required dummy records for a group is $\sum_{e \in \bm{E}_{f, g}}(\tau_{f, g}-O(e))$, and the threshold $\tau_{f, g}=\max\{O(e)\}_{e \in \bm{E}_{f, g}}$.
When the records are divided into more groups, fewer elements will be included in each group.
As a result, the occurrence of the searched element tends to be closer to $\tau_{f, g}$, and then fewer dummy records are required for its padding.
Thus, the result decryption time decreases with the increase of the group number.
In the tested dataset, the elements have very close occurrences, which ranges from $1$ to $41$.
The number of matched dummy records are $9$, $9$, $2$, $1$, and $0$ when there are $1$, $10$, $100$, $1000$, and $10000$ groups, respectively.
For the dataset with a bigger element occurrence gap, the result decryption time will change more obviously when changing the number of groups.

\subsubsection{End-to-end Latency}
Fig. \ref{Fig:mcdb-user} also shows the end-to-end latency on the user side when issuing a query.
In this test, we did not simulate the network latency, so the end-to-end latency shown here consists of the query encryption time, the nonce blinding time, the searching time and the result decrypting time.
The end-to-end latency is dominated by the nonce blinding and searching times, thus it decreases when increasing the number of groups.
Specifically, the end-to-end query decreases from $2.16$ to $0.0006$ seconds when the number of groups increases from $1$ to $10000$.

In this test, we used one trick to improve the performance.
As described in Algorithm \ref{alg:macdb-search}, the SSS is idle before getting $(IL, EN)$.
Indeed, the IWS can send $IL$ to the SSS first, and then the SSS can pre-compute $temp_{id}=H'(EDB(id, EQ.f) \oplus EQ.e^*)$ while the IWS generating $EN$.
After getting $EN$, the SSS just needs to check if $temp_{id}=EN(id).w$.
By computing $(w, t)$ and $temp$ simultaneously, the user can get the search result more efficiently.
In this test, the SSS computed $t$ simultaneously when the IWS generating $EN$.
Note that the shuffle operation does not affect the end-to-end latency on the user side since it is performed after returning search results to users.

\subsection{Insert and Delete Queries}
\begin{figure}
  \centering
  \includegraphics[width=.32\textwidth]{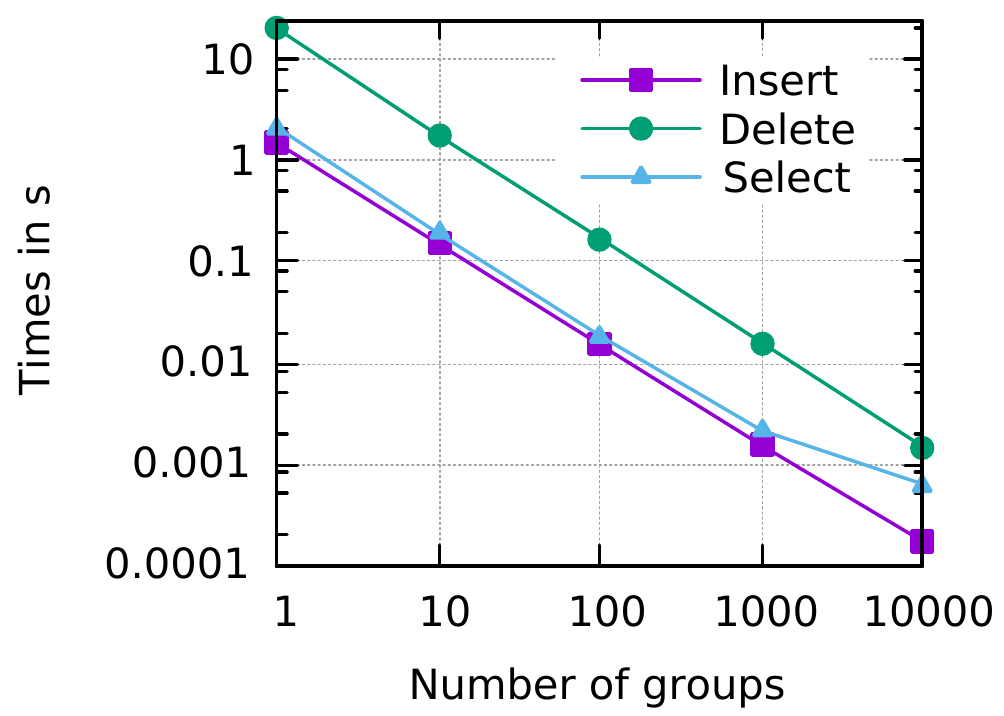}\\
  \caption{The execution times of the insert, delete and select queries with different numbers of groups}
  \label{Fig:mcdb-insert}
\end{figure}
\begin{figure}
  \centering
  \includegraphics[width=.3\textwidth]{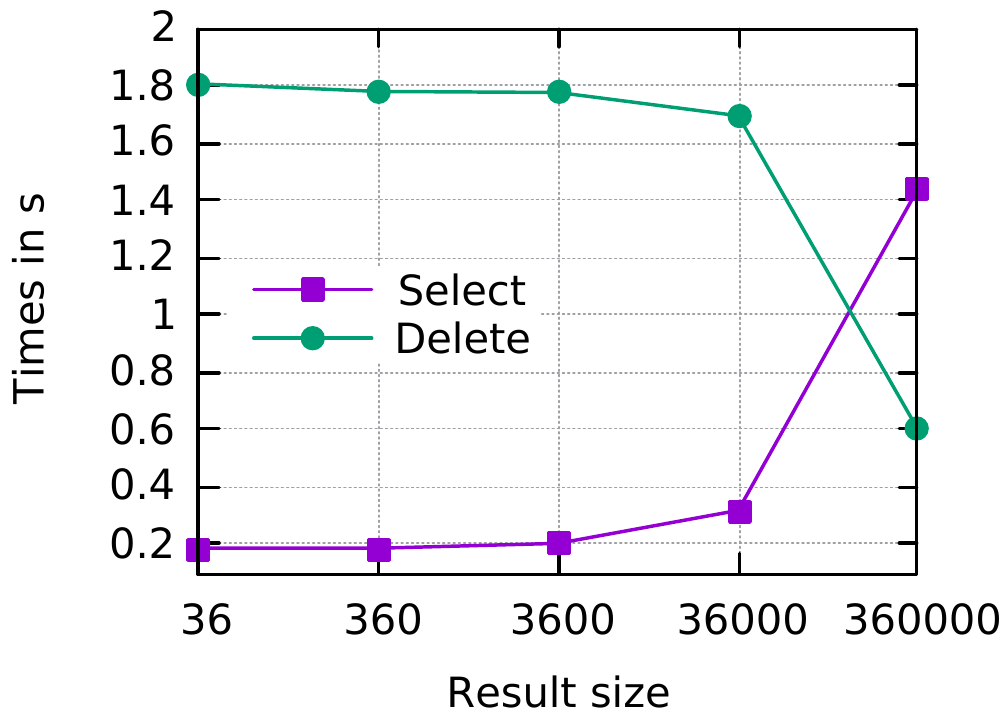}\\
  \caption{The execution times of the insert, delete, and select queries with different result sizes}
  \label{Fig:mcdb-insert-result}
\end{figure}
Since \framework is a dynamic SE scheme, we also tested its performance for insert and delete queries.

In Fig. \ref{Fig:mcdb-insert}, we show the execution times of insert and delete when changing the number of groups\footnote{The times taken by the $PreShuffle$ and $Shuffle$ algorithms are not included.}.
Moreover, we take the end-to-end latency of select queries as the baseline.
Fig. \ref{Fig:mcdb-insert} shows both insert and delete queries execute faster when there are more groups.
For insert queries, as mentioned in Section \ref{subsec:update}, $W=\max\{\gamma_{f}\}_{1 \leq f \leq F}$ dummy records should be inserted when inserting a real record.
Thus, the performance of insert queries is affected by the number of elements in involved groups.
When the database is divided into more groups, the fewer elements will be included in each group.
In this experiment, when there are 1, 10, 100, 1000, and 10000 groups, the user has to generate and encrypt 99996, 10000, 1000, 100, and 10 dummy records, respectively.
Specifically, when there is only one group, \framework takes only around $1.5$ seconds to encrypt and insert $99997$ records, which is slightly better than select queries.

For delete queries, \framework first performs the search operation to get the matched records, and then turn them to dummy by changing the tags.
Moreover, to hide the access pattern from the RSS, the user also needs to change the tags of all the other searched records.
The more groups, the fewer records should be searched, and the fewer tags should be changed.
Therefore, the performance of delete queries also gets better when there are more groups.
However, comparing with select queries, delete queries takes much longer time to change the tags.
Specifically, it takes around $20$ seconds to execute a delete query when there is only one group.

We also tested how the result size affects the performance of select and delete queries.
For this test, we divided the database into 10 groups, and the searched group contains $360000$ records.
Moreover, we manually changed the data distribution in the group to be searched to ensure that we can execute queries matching $36$, $360$, $3600$, $36000$, $360000$ records.
From Fig. \ref{Fig:mcdb-insert-result}, we can see that the performance of delete queries is better when the result size is bigger.
The reason is that tags of matched records are processed in a much efficient way than unmatched records.
Specifically, as mentioned in Section \ref{subsec:update}, the user directly changes the tags of matched records to random strings.
However, for each unmatched record, the user first has to detect if it is dummy or not, and then update their tags accordingly.
When all the searched records match the delete query, it takes only $0.6$ seconds to turn them to dummy.
Nonetheless, select queries take longer time when more records matching the query, since there are more records should be processed on the user side.

%% file: sections/conclusion.tex
\section{Conclusions and Future Directions}
\label{sec:conclusion}
In this work, we presented the leakage profile definitions for searchable encrypted relational databases, and investigated the leakage-based attacks proposed int the literature.
We also proposed \framework, a dynamic searchable encryption scheme for multi-cloud outsourced databases.
\framework does not leak information about search, access, and size patterns.
It also achieves both forward and backward privacy, where the CSPs cannot reuse cached queries for checking if new records have been inserted or if records have been deleted.
\framework has minimal leakage, making it resistant to exiting leakage-based attacks. 

As future work, first we will do our performance analysis by deploying the scheme in the real multi-cloud setting.
Second, we will try to address the limitations of \framework.
Specifically, \framework protects the search, access, and size patterns from the CSPs.
However, it suffers from the collusion attack among CSPs.
In \framework, the SSS knows the search result for each query, and the other two knows how the records are shuffled and re-randomised.
If the SSS colludes with the IWS or RSS, they could learn the search and access patterns.
We will also consider the techniques to defend the collusion attack among CSPs.
Moreover, in this work, we assume all the CSPs are honest.
Yet, in order to learn more useful information, the compromised CSPs might not behave honestly as assumed in the security analysis.
For instance, the SSS might not search all the records indexed by $IL$, and the RSS might not shuffle the records properly.
In the future, we will give a mechanism to detect if the CSPs honestly follow the designated protocols.

%% file: main_IEEE.bbl
\begin{thebibliography}{10}

\bibitem{CuiAGR17}
S.~Cui, M.~R. Asghar, S.~D. Galbraith, and G.~Russello, ``{P-McDb}:
  Privacy-preserving search using multi-cloud encrypted databases,'' in {\em
  {CLOUD} 2017}, pp.~334--341, {IEEE} Computer Society, 2017.

\bibitem{GunawiHLPDAELLM14}
H.~S. Gunawi, M.~Hao, T.~Leesatapornwongsa, T.~Patana{-}anake, T.~Do,
  J.~Adityatama, K.~J. Eliazar, A.~Laksono, J.~F. Lukman, V.~Martin, and A.~D.
  Satria, ``What bugs live in the cloud? {A} study of 3000+ issues in cloud
  systems,'' in {\em {SoCC} 2014}, pp.~7:1--7:14, {ACM}, 2014.

\bibitem{dropboxleaks}
{Tech Republic}, ``{Dropbox and Box leak files in security through obscurity
  nightmare}.''
  \url{https://www.techrepublic.com/article/dropbox-and-box-leak-files-in-security-through-obscurity-nightmare},
  2014.
\newblock Last accessed: July 8, 2019.

\bibitem{verizonreport}
{Verrizon}, ``{2019 Data Breach Investigations Report}.''
  \url{https://enterprise.verizon.com/resources/reports/dbir/}, 2019.
\newblock Last accessed: July 8, 2019.

\bibitem{Song:2000:Practical}
D.~X. Song, D.~Wagner, and A.~Perrig, ``Practical techniques for searches on
  encrypted data,'' in {\em {S{\&}P} 2000}, pp.~44--55, {IEEE} Computer
  Society, 2000.

\bibitem{Asghar:2013:CCSW}
M.~R. Asghar, G.~Russello, B.~Crispo, and M.~Ion, ``Supporting complex queries
  and access policies for multi-user encrypted databases,'' in {\em {CCSW}
  2013}, pp.~77--88, {ACM}, 2013.

\bibitem{Curtmola:2006:Searchable}
R.~Curtmola, J.~A. Garay, S.~Kamara, and R.~Ostrovsky, ``Searchable symmetric
  encryption: improved definitions and efficient constructions,'' in {\em {CCS}
  2006}, pp.~79--88, {ACM}, 2006.

\bibitem{Popa:2011:Cryptdb}
R.~A. Popa, C.~M.~S. Redfield, N.~Zeldovich, and H.~Balakrishnan, ``{CryptDB}:
  protecting confidentiality with encrypted query processing,'' in {\em {SOSP}
  2011}, pp.~85--100, {ACM}, 2011.

\bibitem{Sarfraz:2015:DBMask}
M.~I. Sarfraz, M.~Nabeel, J.~Cao, and E.~Bertino, ``{DBMask}: Fine-grained
  access control on encrypted relational databases,'' in {\em {CODASPY} 2015},
  pp.~1--11, {ACM}, 2015.

\bibitem{Islam:2012:Access}
M.~S. Islam, M.~Kuzu, and M.~Kantarcioglu, ``Access pattern disclosure on
  searchable encryption: Ramification, attack and mitigation,'' in {\em {NDSS}
  2012}, The Internet Society, 2012.

\bibitem{Naveed:2015:Inference}
M.~Naveed, S.~Kamara, and C.~V. Wright, ``Inference attacks on
  property-preserving encrypted databases,'' in {\em {SIGSAC} 2015},
  pp.~644--655, {ACM}, 2015.

\bibitem{Cash:2015:leakage}
D.~Cash, P.~Grubbs, J.~Perry, and T.~Ristenpart, ``Leakage-abuse attacks
  against searchable encryption,'' in {\em {SIGSAC} 2015}, pp.~668--679, {ACM},
  2015.

\bibitem{Zhang:2016:All}
Y.~Zhang, J.~Katz, and C.~Papamanthou, ``All your queries are belong to us: The
  power of file-injection attacks on searchable encryption,'' in {\em {USENIX}
  Security 2016}, pp.~707--720, USENIX Association, 2016.

\bibitem{Kellaris:ccs16:Generic}
G.~Kellaris, G.~Kollios, K.~Nissim, and A.~O'Neill, ``Generic attacks on secure
  outsourced databases,'' in {\em {SIGSAC} 2016}, pp.~1329--1340, {ACM}, 2016.

\bibitem{Abdelraheem:eprint17:record}
M.~A. Abdelraheem, T.~Andersson, and C.~Gehrmann, ``Inference and
  record-injection attacks on searchable encrypted relational databases,'' {\em
  {IACR} Cryptology ePrint Archive}, p.~24, 2017.

\bibitem{BostMO17}
R.~Bost, B.~Minaud, and O.~Ohrimenko, ``Forward and backward private searchable
  encryption from constrained cryptographic primitives,'' in {\em {CCS} 2017},
  pp.~1465--1482, {ACM}, 2017.

\bibitem{ZuoSLSP18}
C.~Zuo, S.~Sun, J.~K. Liu, J.~Shao, and J.~Pieprzyk, ``Dynamic searchable
  symmetric encryption schemes supporting range queries with forward (and
  backward) security,'' in {\em {ESORICS} 2018}, pp.~228--246, Springer, 2018.

\bibitem{ChamaniPPJ18}
J.~G. Chamani, D.~Papadopoulos, C.~Papamanthou, and R.~Jalili, ``New
  constructions for forward and backward private symmetric searchable
  encryption,'' in {\em {CCS} 2018}, pp.~1038--1055, {ACM}, 2018.

\bibitem{AmjadKM19}
G.~Amjad, S.~Kamara, and T.~Moataz, ``Forward and backward private searchable
  encryption with {SGX},'' in {\em {EuroSec} 2019}, pp.~4:1--4:6, {ACM}, 2019.

\bibitem{Arasu:CIDR13:Cipherbase}
A.~Arasu, S.~Blanas, K.~Eguro, R.~Kaushik, D.~Kossmann, R.~Ramamurthy, and
  R.~Venkatesan, ``Orthogonal security with cipherbase,'' in {\em {CIDR} 2013},
  www.cidrdb.org, 2013.

\bibitem{Hahn:2014:Searchable}
F.~Hahn and F.~Kerschbaum, ``Searchable encryption with secure and efficient
  updates,'' in {\em {SIGSAC} 2014}, pp.~310--320, {ACM}, 2014.

\bibitem{Poddar:arx:eprint16}
R.~Poddar, T.~Boelter, and R.~A. Popa, ``{Arx}: {A} strongly encrypted database
  system,'' {\em {IACR} Cryptology ePrint Archive}, p.~591, 2016.

\bibitem{BonehCOP04}
D.~Boneh, G.~D. Crescenzo, R.~Ostrovsky, and G.~Persiano, ``Public key
  encryption with keyword search,'' in {\em {EUROCRYPT} 2004}, pp.~506--522,
  Springer, 2004.

\bibitem{Fisch:SP15:BlindSeer}
B.~A. Fisch, B.~Vo, F.~Krell, A.~Kumarasubramanian, V.~Kolesnikov, T.~Malkin,
  and S.~M. Bellovin, ``Malicious-client security in blind seer: {A} scalable
  private {DBMS},'' in {\em {SP} 2015}, pp.~395--410, {IEEE} Computer Society,
  2015.

\bibitem{Goldreich:1996:SPS}
O.~Goldreich and R.~Ostrovsky, ``Software protection and simulation on
  {Oblivious RAMs},'' {\em J. {ACM}}, pp.~431--473, 1996.

\bibitem{Stefanov:2013:PathORAM}
E.~Stefanov, M.~van Dijk, E.~Shi, C.~W. Fletcher, L.~Ren, X.~Yu, and
  S.~Devadas, ``Path {ORAM}: An extremely simple oblivious {RAM} protocol,'' in
  {\em {SIGSAC} 2013}, pp.~299--310, {ACM}, 2013.

\bibitem{Paillier:1999:Public}
P.~Paillier, ``Public-key cryptosystems based on composite degree residuosity
  classes,'' in {\em {EUROCRYPT} 1999}, pp.~223--238, Springer, 1999.

\bibitem{Gentry:2009:FHE}
C.~Gentry, ``Fully homomorphic encryption using ideal lattices,'' in {\em
  {STOC} 2009}, pp.~169--178, {ACM}, 2009.

\bibitem{Samanthula:2014:Privacy}
B.~K. Samanthula, W.~Jiang, and E.~Bertino, ``Privacy-preserving complex query
  evaluation over semantically secure encrypted data,'' in {\em {ESORICS}
  2014}, pp.~400--418, Springer, 2014.

\bibitem{Ishai:2016:Private}
Y.~Ishai, E.~Kushilevitz, S.~Lu, and R.~Ostrovsky, ``Private large-scale
  databases with distributed searchable symmetric encryption,'' in {\em
  {CT-RSA} 2016}, pp.~90--107, Springer, 2016.

\bibitem{Tu:PVLDB13:monomi}
S.~Tu, M.~F. Kaashoek, S.~Madden, and N.~Zeldovich, ``Processing analytical
  queries over encrypted data,'' {\em {PVLDB}}, pp.~289--300, 2013.

\bibitem{Papadimitriou:usenix2016:Seabed}
A.~Papadimitriou, R.~Bhagwan, N.~Chandran, R.~Ramjee, A.~Haeberlen, H.~Singh,
  A.~Modi, and S.~Badrinarayanan, ``Big data analytics over encrypted datasets
  with seabed,'' in {\em {OSDI} 2016}, {USENIX} Association, 2016.

\bibitem{Pappas:BlindSeer:SP14}
V.~Pappas, F.~Krell, B.~Vo, V.~Kolesnikov, T.~Malkin, S.~G. Choi, W.~George,
  A.~D. Keromytis, and S.~Bellovin, ``Blind seer: {A} scalable private
  {DBMS},'' in {\em {SP} 2014}, pp.~359--374, {IEEE} Computer Society, 2014.

\bibitem{KirkpatrickGV83}
S.~Kirkpatrick, D.~G. Jr., and M.~P. Vecchi, ``Optimization by simmulated
  annealing,'' {\em Science}, vol.~220, no.~4598, pp.~671--680, 1983.

\bibitem{eronemail:2017}
``Enron email dataset.'' \url{http://www.enronemail.com}.
\newblock Last accessed: July 8, 2019.

\bibitem{Abdelraheem:2017:Seachable}
M.~A. Abdelraheem, T.~Andersson, and C.~Gehrmann, ``Searchable encrypted
  relational databases: Risks and countermeasures,'' in {\em {DPM} 2017 and
  {CBT} 2017}, pp.~70--85, Springer, 2017.

\bibitem{RightScale:2016:report}
``{RightScale 2019}.'' \url{https://www.rightscale.com/lp/state-of-the-cloud}.
\newblock Last accessed: July 8, 2019.

\bibitem{Hoang:2016:practical}
T.~Hoang, A.~A. Yavuz, and J.~Guajardo, ``Practical and secure dynamic
  searchable encryption via oblivious access on distributed data structure,''
  in {\em {ACSAC} 2016}, pp.~302--313, {ACM}, 2016.

\bibitem{Stefanov:CCS2013:Multi-cloud}
E.~Stefanov and E.~Shi, ``Multi-cloud oblivious storage,'' in {\em {SIGSAC}
  2013}, pp.~247--258, {ACM}, 2013.

\bibitem{Asghar:Espoon:IRC13}
M.~R. Asghar, M.~Ion, G.~Russello, and B.~Crispo, ``${ESPOON}_{ERBAC}$:
  Enforcing security policies in outsourced environments,'' {\em Computers {\&}
  Security}, pp.~2--24, 2013.

\bibitem{Knuth73}
D.~E. Knuth, {\em The Art of Computer Programming, Volume {III:} Sorting and
  Searching}.
\newblock Addison-Wesley, 1973.

\bibitem{Bost:2017:forward}
R.~Bost, B.~Minaud, and O.~Ohrimenko, ``Forward and backward private searchable
  encryption from constrained cryptographic primitives,'' in {\em {CCS} 2017},
  pp.~1465--1482, {ACM}, 2017.

\bibitem{CashJJJKRS14}
D.~Cash, J.~Jaeger, S.~Jarecki, C.~S. Jutla, H.~Krawczyk, M.~Rosu, and
  M.~Steiner, ``Dynamic searchable encryption in very-large databases: Data
  structures and implementation,'' in {\em {NDSS} 2014}, The Internet Society,
  2014.

\bibitem{KamaraPR12}
S.~Kamara, C.~Papamanthou, and T.~Roeder, ``Dynamic searchable symmetric
  encryption,'' in {\em {CCS} 2012}, pp.~965--976, {ACM}, 2012.

\bibitem{TPC:2017:h}
``{TPC-H}.''
\newblock Last accessed: July 8, 2019.

\end{thebibliography}
